\documentclass{article}

\usepackage[round]{natbib}
\usepackage{times}
\usepackage[margin=1in]{geometry}

\usepackage{graphicx}
\usepackage{amsmath,amsfonts,amssymb}
\usepackage{amsthm}
\usepackage{thmtools}
\usepackage{thm-restate}
\usepackage{enumitem}
\usepackage{doi}
\usepackage{hyperref}
\usepackage[capitalise]{cleveref}
\usepackage{comment}
\usepackage{multirow}

\usepackage{booktabs}
\usepackage{array}
\usepackage{tabularx}

\declaretheorem[name=Theorem,numberwithin=section]{theorem}
\declaretheorem[name=Lemma,sibling=theorem]{lemma}
\declaretheorem[name=Claim,sibling=theorem]{claim}
\declaretheorem[name=Corollary,sibling=theorem]{corollary}
\declaretheorem[name=Definition,style=definition,sibling=theorem]{definition}

\newcommand{\RR}{\mathbb{R}}
\newcommand{\Rd}{\RR^d}
\newcommand{\norm}[1]{\left\lVert#1\right\rVert}
\newcommand{\abs}[1]{\left\lvert#1\right\rvert}
\newcommand{\ip}[2]{\left\langle #1,\, #2 \right\rangle}
\newcommand{\Xspace}{\mathcal{H}}
\newcommand{\mean}{\mu}           % population mean / optimal center
\newcommand{\smean}{\hat{\mu}}    % sample mean
\providecommand{\eqdef}{:=}
\DeclareMathOperator{\poly}{poly}
\DeclareMathOperator{\dist}{dist}

\DeclareMathOperator{\cost}{cost}

\DeclareMathOperator{\costm}{cost^m}     % matched cost
\DeclareMathOperator{\OPT}{OPT}          % optimal k-mean cost
\DeclareMathOperator{\OPTQ}{\widehat{OPT}}          % empirical k-mean cost for Q

\newcommand{\Vset}{\mathcal{V}}
\newcommand{\lift}{u}
\newcommand{\Gspace}{\mathcal{G}}

\DeclareMathOperator{\Var}{Var}
\newcommand{\E}{\mathbb E}

\newcommand{\indic}[1]{ \mathbf{1}_{\{#1\}} } % notation for indicator variables

\newcommand{\sol}{\mathcal{S}_\gamma}
\newcommand{\union}{\sol(Q) \cup \sol(P)}
\newcommand{\event}{\mathcal{E}_Q}
\newcommand{\Gams}{\Gamma_\gamma}

\title{Coresets for $k$-Means: When Does Uniform Sampling Suffice?}

\author{Amir Carmel%
    \thanks{Email: \texttt{amir6423@gmail.com}
    }
    \qquad
    Robert Krauthgamer%
        \thanks{The Harry Weinrebe Professorial Chair of Computer Science.
            Work partially supported by the Israel Science Foundation grant \#1336/23.
            Email: \texttt{robert.krauthgamer@weizmann.ac.il}
        }
    \\  Weizmann Institute of Science
}
\date{}

\begin{document}
\maketitle

\begin{abstract}
Data summarization techniques, such as coresets, are essential for scaling clustering algorithms to massive datasets. 
We study uniform sampling, which is an extremely simple and ultra-efficient method to construct coresets,
in the context of weak, stable, and strong coresets for the $k$-means problem in Hilbert space.
We begin with the fundamental case $k=1$, 
and establish matching upper and lower bounds on the sample complexity of stable and weak coresets, 
finding that the two notions require different sample complexities. 
As a consequence, we obtain tight bounds for stable $1$-median coresets in $\ell_1$ and $\ell_2$ (and thus also Kendall--tau and Jaccard metrics), removing the dimension dependence of Carmel and Krauthgamer (ICLR 2026). 
Our analysis is substantially simpler than prior work and reduces to bounding the deviation of the sample mean.
For the general $k$-means problem on $\rho$-separated instances, a well-studied structural assumption on the input,
we show that uniform sampling yields a stable coreset for all $O(1)$-approximate solutions, 
and a strong coreset if the instance additionally has bounded kurtosis.
In all these results, the sample complexity is independent of the dimension and the instance size.
\end{abstract}

% ==================================================================
\section{Introduction}\label{section:introduction}

Center-based clustering seeks to partition a set of points into groups of similar
points and further assign a single representative center to each partition. Formally, in the $(k,z)$-clustering
problem, the input is a set $P$ of $n$ points in a metric space $(\mathcal{X},\dist)$ and an
integer $k\ge 1$, and the goal is to find a set of $k$ centers $C\subseteq\mathcal{X}$
minimizing $\cost_z(C,P):=\frac{1}{n}\sum_{p\in P}\min_{c\in C}\dist(p,c)^z$. We focus on $z\in [1,2]$ where the cases $z=1$
and $z=2$ are the classical $k$-median and $k$-means problems, respectively. The
special case $k=1$, namely the $1$-median and $1$-mean problems, is interesting in
its own right and serves as a basic building block for the general problem.

Modern datasets are massive, and this motivates data summarization techniques that let
downstream algorithms run on a smaller input rather than the original one. A coreset
is a small, reweighted subset of the input that provably preserves the relevant
properties of the clustering objective,
so that solving the problem (optimally or
near optimally) on the coreset gives a near-optimal solution for the original input.
Coresets have been the subject of an extensive line of work, offering tunable
tradeoffs between size and accuracy, and are the central tool we study here.
(Since our results focus on coresets constructed by uniform sampling, we present the coreset definitions without weights.)

The prevailing coreset notion is the \emph{strong $\epsilon$-coreset}, a subset $Q\subseteq P$ such that for every
candidate center set $C \subseteq \mathcal{X}$ of size $k$,
\begin{equation}\label{eq:strong}
   \cost_z(C,Q)\in(1\pm\epsilon)\cost_z(C,P) .
\end{equation}
It is well known that constructing a strong coreset requires reading
the entire input $P$, as a single far point (outlier) can dominate the cost,
and not including it in $Q$ would violate the guarantee. 
At the scale of
modern datasets, reading the entire input is prohibitive and thus other notions were suggested. 
A \emph{weak $(\epsilon,\eta)$-coreset} is such a notion, that in contrast, only
requires that near-optimal solutions on $Q$ remain near-optimal on $P$. Writing
$\OPT_z(Q):=\min_{C'\subseteq\mathcal{X},\abs{C'}=k}\cost_z(C',Q)$, and similarly for
$P$,
it formally requires that for all $C\subseteq\mathcal{X}$ of size $k$,
\begin{equation}\label{eq:weak}
  \text{if }
  \cost_z(C, Q) \leq (1+\epsilon)\,\OPT_z(Q), \quad \text{then }\cost_z(C, P) \leq (1+\eta)\,\OPT_z(P).
\end{equation}
While this guarantees the most basic property one would expect from a coreset, it says
nothing about candidates $C$ that are not near-optimal on $Q$. This is a serious
limitation whenever the centers are constrained to lie in some subset
$X\subset\mathcal{X}$, since an optimal solution over $X$ need not be near-optimal over
all of $\mathcal{X}$ (we give concrete examples of such constraints below).

Recently,~\citet{carmel2025stablecoresetsunleashingpower}
introduced an intermediate notion, the \emph{stable $(\epsilon,\eta)$-coreset}, which strengthens the weak coreset guarantee by approximately preserving the relative order of costs between \emph{every} pair of candidate solutions. Formally it requires that, for all $C_1,C_2\subseteq\mathcal{X}$ of size $k$,
\begin{equation}\label{eq:stable}
  \text{if }
  \cost_z(C_1, Q) \leq (1+\epsilon)\cost_z(C_2,Q), \quad \text{then }\cost_z(C_1, P) \leq (1+\eta) \cost_z(C_2,P).
\end{equation}
This guarantee transfers to every submetric of $\mathcal{X}$, and hence to every variant
of the problem in which the centers are restricted to a subset
$X\subset\mathcal{X}$, without even needing to know $X$ in advance, yet it
is permissive enough to be achieved by uniform
sampling without reading the entire input.
Concretely, it is shown in~\citep{carmel2025stablecoresetsunleashingpower} 
that uniform sampling produces a stable coreset for $1$-median in $\ell_1$, 
and since tree metrics, Hamming, Kendall-tau,
and Jaccard distance are all submetrics of $\ell_1$, this transfer-to-submetric property
immediately unlocks coreset results for all of these, which were not previously studied.
This property also captures constraints on the solution itself, 
such as the discrete centers (medoids) variant, where centers must be input points, 
or kernel-based clustering, where the data points are images of a feature map 
and also the centers are constrained to lie in that map's image~\citep{FeldmanMS07, DBLP:journals/ml/JiangKLZ24, carmel2026optimalstablecoresetsgeometric}.

Even though $(k,z)$-clustering is NP-hard in general (in fact
APX-hard~\citep{DBLP:conf/soda/Cohen-AddadSL22}), it is solved routinely in practice,
and this apparent tension has motivated a large body of work isolating structural
assumptions under which the problem becomes
tractable (see~\citep{BenDavid2015,DBLP:books/cu/20/R2020} for surveys). One
recurring assumption is $\rho$-separation, introduced by~\citet{rabani}, which requires that the optimal cost with
$k$ centers is substantially smaller than with $k-1$ centers, 
formalizing the intuition behind the elbow method used in practice to choose $k$. 
Another common assumption is that each cluster has bounded dispersion, in the sense of not having a heavy tail,
for instance because it arises from Gaussian or sub-exponential distributions;
we formalize this via a bound on the kurtosis, which is the fourth moment of each cluster relative to the square of its variance.

All our coreset results are obtained by sampling points from $P$ uniformly. Uniform sampling is a cheap, data-oblivious approach that requires no preprocessing of the data, and is therefore particularly suitable for streaming and distributed settings. This distinguishes it from sensitivity sampling and other strong coreset constructions, which inherently require processing the entire input. 
Our goal is to understand which guarantees can be obtained without processing the entire
input. This stands in contrast to much of the literature, which has generally dismissed
uniform sampling as a viable construction for (strong) coresets, in both theory and
practice~\citep[Appendix~B]{DBLP:conf/esa/SchwiegelshohnS22Evaluation}.
Thus, the structural assumptions we use below for general $k$ should be viewed as investigating the conditions under which uniform sampling succeeds.
As discussed in~\citep{mop04,DBLP:conf/icml/HuangJL23},
one structural requirement for such success is some form of balancedness of the optimal clusters, 
since a cluster of small size is likely to be completely missed by the sample. 
This is in line with practice, where a very small
fraction of points is routinely treated as outliers rather than as an extra cluster. 
We adopt a standard formalization called \emph{$\beta$-balancedness}, 
which requires that every cluster in the optimal solution contains at least a $\beta/k$-fraction of the points.
Our objective is to find what uniform sampling can really achieve despite these apparent limitations.

\subsection{Our contribution}\label{section:contribution}

Our results are obtained in Hilbert space, which from now on we denote by $\Xspace = \mathcal{X}$.
Our first technical result shows that uniform sampling yields a stable coreset for the $1$-mean problem in $\Xspace$, with a sample size that is independent of the dimension.
Previously, no bound was known for this problem. 

\begin{restatable}[Stable coreset for $1$-mean]{theorem}{mainthm} 
\label{thm:stable}
Let $\Xspace$ be a Hilbert space and let $P \subseteq \Xspace$ be finite. 
Then for every $\epsilon \in (0,1/10)$ and $\delta >0$,
a uniform sample of (sufficiently large) size $O(\epsilon^{-2}\delta^{-1})$ from $P$ is, with probability at least $1-\delta$,
a stable $(\delta\epsilon,\, 8\epsilon)$-coreset for the $1$-mean problem.
\end{restatable}

Since for every $p\in[1,2]$ the space $\ell_p^d$ ($\Rd$ endowed with the $\ell_p$ norm) embeds
isometrically into $\ell_2$-squared~\citep{schoenberg1938,BDCK67}, 
\Cref{thm:stable} immediately yields the same guarantee also for $1$-median in $\ell_p^d$.

\begin{restatable}[Stable coreset for $1$-median]{corollary}{maincorollary}\label{cor:median}
Fix $p \in [1, 2]$ and let $P \subset \Rd$ be finite. 
Then for every $\epsilon \in (0,1/10)$ and $\delta >0$, 
a uniform sample of (sufficiently large) size $O(\epsilon^{-2}\delta^{-1})$ from $P$ is,
with probability at least $1-\delta$, a stable
$(\delta\epsilon,\, 8\epsilon)$-coreset for $1$-median in $\ell_p^d$.
\end{restatable}

This corollary is perhaps our most important contribution for $k=1$,
as it improves the two known bounds on stable coresets for $1$-median and has several
applications (see also Appendix~\ref{sec:applications} for additional implications of
Theorem~\ref{thm:stable}).
The endpoints $p=1$ and $p=2$ recover $1$-median in $\ell_1^d$ and $\ell_2^d$,
and the same conclusion holds for every metric that embeds isometrically into $\ell_1$,
including Hamming, Kendall-tau, Spearman footrule, Jaccard, and tree metrics.
These two prior results are as follows. The first~\citep{carmel2025stablecoresetsunleashingpower}
gives a bound for $\ell_1^d$ that depends logarithmically on the dimension $d$, proved
using VC-dimension machinery. The second~\citep{carmel2026optimalstablecoresetsgeometric}
employs an iterative size reduction while exploiting the construction
of~\citep{carmel2025stablecoresetsunleashingpower}, and its bound for $\ell_2^d$ carries
a $\log(1/\epsilon)$ gap that we close.
\Cref{tab:comparison} summarizes these and related works on $1$-median coresets.

\begin{table}[t]
\centering
\begin{tabular}{|l|l|l|l|}
\hline
Metric & Coreset type & Coreset size & Reference \\
\hline
\multirow{3}{*}{$\ell_2^d$} & \multirow{3}{*}{weak}
  & $O(\epsilon^{-4})$              & \citet{DBLP:conf/nips/Cohen-AddadSS21} \\
  &  & $O(\epsilon^{-3})$              & \citet{DBLP:conf/icml/HuangJL23}$^\dagger$ \\
  &  & $O(\epsilon^{-2}\log \epsilon^{-1})$ & \citet{WY24} \\
\hline
\multirow{2}{*}{$\ell_2^d$} & \multirow{2}{*}{stable}
  & $O(\epsilon^{-2}\log (d/\epsilon))$  & \citet{carmel2025stablecoresetsunleashingpower} \\
  &  & $O(\epsilon^{-2}\log \epsilon^{-1})$ & \citet{carmel2026optimalstablecoresetsgeometric} \\
\hline
$\ell_1^d$ & stable & $O(\epsilon^{-2}\log d)$        & \citet{carmel2025stablecoresetsunleashingpower} \\
\hline
$\ell_p^d,\ p\in[1,2]$ & stable & $O(\epsilon^{-2})$              & \Cref{cor:median} \\
\hline
\end{tabular}
\caption{Uniform-sampling coreset sizes for stable and weak $(\epsilon, O(\epsilon))$-coresets
for the $1$-median problem.
$^\dagger$The cited work is primarily concerned with the more general $k$-median problem.}
\label{tab:comparison}
\end{table}

Our key idea is to aim for a stronger result,
namely, a coreset for $1$-mean in Hilbert space, 
which implies a coreset for $1$-median in $\ell_1^d$ or $\ell_2^d$,
by embedding these spaces isometrically into $\ell_2$-squared
and applying the submetric property of stable coresets.
Surprisingly, our proof is much simpler
as it avoids the VC-dimension machinery entirely,
and it furthermore identifies the actual source of difficulty in 
analyzing stable and weak coresets constructed via uniform sampling.
Since every stable coreset is also a weak coreset, 
our results immediately improve upon all previous weak coreset constructions summarized in Table~\ref{tab:comparison}, with substantially simpler analysis.

To place our stable coreset results in context, 
the next theorem states the corresponding bound for weak coresets. 
It was essentially proved in~\citep{DBLP:conf/compgeom/InabaKI94, DBLP:journals/corr/abs-2111-03046},
albeit only for weak $(0,\epsilon)$-coresets,
and for completeness we provide a short proof of the extended bound in Appendix~\ref{appendix:weak}. 
We remark that weak $(\epsilon,O(\epsilon))$-coresets were established in~\citep{DBLP:conf/nips/Cohen-AddadSS21} 
also for a more general problem called Euclidean power mean, and were later improved in~\citep{WY24}.

\begin{restatable}[Weak coreset for $1$-mean]{theorem}{weakthm}\label{thm:weak}
Let $\Xspace$ be a Hilbert space and let $P \subseteq \Xspace$ be finite. 
Then for every $\epsilon \in (0,1/10)$ and $\delta >0$, 
a uniform sample of (sufficiently large) size $O(\epsilon^{-1}\delta^{-1})$ from $P$ is, with probability at least $1-\delta$,
a weak $(\delta\epsilon,\, 8\epsilon)$-coreset for the $1$-mean problem.
\end{restatable}

We complement \Cref{thm:stable} and \Cref{thm:weak} with matching lower bounds. 
For weak coresets, there is already a lower bound for every power $z \geq 2$ in~\citep{DBLP:conf/nips/Cohen-AddadSS21}, and thus item \ref{item:weak} below was already known. 

\begin{restatable}{theorem}{lowerbounds}\label{thm:lowerbounds}
For every $\epsilon,\delta\in(0,1/20)$, there exists a finite set $P\subset\RR$ such that if a uniform sample $Q\subseteq P$ of size $s$ is, with probability at least $1-\delta$,
\begin{enumerate}[label=(\roman*),nosep]
    \item\label{item:weak} a weak $(0,\epsilon)$-coreset for the $1$-mean problem, then $s=\Omega(\epsilon^{-1}\delta^{-1})$;
    \item\label{item:stable} a stable $(0,\epsilon)$-coreset for the $1$-mean problem, then $s=\Omega(\epsilon^{-2}\delta^{-1})$.
\end{enumerate}
\end{restatable}

Together, \Cref{thm:stable,thm:weak,thm:lowerbounds} provide a direct comparison between the two notions for the $1$-mean problem in Hilbert space, revealing a clear separation: stable coresets require $\Theta(\epsilon^{-2}\delta^{-1})$ samples, whereas weak coresets require only $\Theta(\epsilon^{-1}\delta^{-1})$. 
In contrast, there is no such separation for the $1$-median problem, as both notions are known to require, 
up to polylogarithmic factors, $\Theta(\epsilon^{-2})$ samples for constant success probability~\citep{carmel2026optimalstablecoresetsgeometric}.

We then consider the case of general $k \geq 2$, a regime where uniform-sampling constructions are scarce. 
Prior stable-coreset results are restricted to $k=1$~\citep{carmel2025stablecoresetsunleashingpower, carmel2026optimalstablecoresetsgeometric},
and the closest prior work for $k\ge2$ is a weak $k$-median coreset of~\citep{DBLP:conf/icml/HuangJL23},
and we compare to it further below. 
For $\gamma\ge1$, we consider the candidate center sets
$\sol(P) := \{C :\ \cost(C,P)\le\gamma\cdot\OPT_2(P)\}$ and similarly
$\sol(Q) := \{C :\ \cost(C,Q)\le\gamma\cdot\OPT_2(Q)\}$.
We think of $\gamma$ as a small constant, since constant-factor approximations for $(k,z)$-clustering 
are readily available.
We first show that under $\beta$-balancedness and $\rho$-separation assumptions, 
a uniform sample already preserves the cost order (in the stable sense) 
of all solutions whose cost on $P$ or on $Q$ is within factor $\gamma$ of optimal.

\begin{restatable}{theorem}{kmeansstable}\label{thm:kmeans-stable}
Let $\Xspace$ be a Hilbert space and let $k\ge 2$, $\gamma\ge 1$, $\epsilon\in(0,1/75)$, $\delta\in(0,1)$. 
Then for every finite $P\subset \Xspace$ that is $\beta$-balanced for $\beta\in(0,1]$
and $\rho$-separated with $\rho^2\le\beta \delta/(576\gamma k)$,
a uniform sample $Q$ of (sufficiently large) size 
$O\big( {k^2 \gamma}/(\beta\epsilon^2\delta^2) \big)$ 
is, with probability at least $1-\delta$, a stable
$(\delta\epsilon,\,O(\epsilon))$-coreset of $P$ over solutions in $\union$ for the $k$-means problem.
\end{restatable}

Observe that by setting $\gamma=1+\epsilon$, we immediately obtain a weak $(\epsilon,O(\epsilon))$-coreset guarantee. 
The closest prior result, obtained by \citet{DBLP:conf/icml/HuangJL23}, is a weak coreset of size
$O\left(\frac{k^2}{\beta\epsilon^3}\log^2\frac{k}{\beta\epsilon}\log^2\frac1\epsilon\right)$
for $k$-median in $\ell_2^d$, with constant success probability. 
They too require that $P$ is $\beta$-balanced, but their weak coreset guarantee 
is further restricted to candidate solutions that are balanced on the sample $Q$. 
We require instead that the input $P$ is $\rho$-separated, and for constant success
probability our sample complexity matches their bound in $k$ and improves the dependence
on $\epsilon$ from $\epsilon^{-3}$ to $\epsilon^{-2}$, while also giving the stronger
stable guarantee over all $O(1)$-approximate solutions.
Neither their result nor their technique extends to the $k$-means problem, 
whereas ours, stated above for squared-distance objective in Hilbert space, 
implies the Euclidean $k$-median setting as a special case.
Another advantage of our result is that if the input $P$ has in addition bounded kurtosis, 
then the stable guarantee strengthens to a strong one, i.e., it preserves the cost of every candidate center set in $\union$.
This improves over the strong coreset result of~\citep{carmel2025stablecoresetsunleashingpower}, 
which is restricted to $k=1$ and assumes bounded dispersion 
(each cluster diameter relative to its average distance).
Bounded kurtosis is clearly a more relaxed requirement and arguably more typical in practice.

\begin{restatable}{theorem}{kmeansstrong}\label{thm:kmeans-strong}
Let $\Xspace$ be a Hilbert space and let $k\ge 2$, $\gamma\ge 1$, $\epsilon\in(0,1/75)$, $\delta\in(0,1)$.
Then for every finite $P\subset \Xspace$ that has kurtosis bounded by $M>0$, 
is $\beta$-balanced for $\beta\in(0,1]$,
and is $\rho$-separated with $\rho^2\le\beta/(120\gamma k)$, 
a uniform sample $Q$ of (sufficiently large) size $O\big( k^2M/(\beta\epsilon^2\delta) \big)$ satisfies, 
with probability at least $1-\delta$, 
\[
  \forall C\in\union,
  \qquad
  \cost(C,Q) \in (1 \pm \epsilon)\,\cost(C,P).
\]
\end{restatable}

We complement our theoretical results with an experimental evaluation (\Cref{section:experiments}) on real-world datasets, 
comparing uniform sampling against importance-sampling constructions of coreset for both $1$-median and $k$-means. 
We find that uniform sampling matches the empirical accuracy of these more expensive methods, despite requiring several orders of magnitude less construction time and being oblivious to the input.

\subsection{Technical overview}

\paragraph{Upper bound for \texorpdfstring{$k=1$}{k=1} (Theorem~\ref{thm:stable}).}
The key idea is to analyze $1$-mean in Hilbert space, 
rather than $1$-median in $\ell_1^d$ or $\ell_2^d$.
We use the sufficient condition of~\citep{carmel2025stablecoresetsunleashingpower} for a stable coreset, called \emph{relative cost-difference approximation} (RCDA) property: 
for every candidate center set $C$, the difference between its cost and the optimal cost is approximately the same whether measured on $P$ or on $Q$ (see \Cref{def:rcda}).

Now for the $1$-mean problem in Hilbert space $\Xspace$, 
we prove that the RCDA error at a center $c$ factors into 
a center-dependent term and a sample-dependent term $\Delta$, 
which is the deviation of the sample mean from the true mean. 
Since this $\Delta$ does not depend on $c$, a single concentration bound controls it, 
and hence also the RCDA error, simultaneously for all centers. 
Once we establish that a uniform sample $Q$ has the RCDA property (with high probability), 
we just invoke the framework of~\citep{carmel2025stablecoresetsunleashingpower} 
(see \Cref{thm:rcda-to-stable}) to conclude that $Q$ is a stable coreset.

\paragraph{Lower bounds for $k=1$ (Theorem~\ref{thm:lowerbounds}).}
Both bounds use the same one-dimensional instance, where $\frac1n$-fraction of the mass is at a point $D>0$ 
and the rest is at the origin, so the sample sees the rare point only through the number of times it is drawn. 
The separation between the two notions comes from how accurately the sample must reproduce the mean. A weak coreset only needs the sample mean to be within $O\big(\sqrt{\epsilon \OPT_2(P)}\big)$ of the true mean. A stable coreset must preserve the cost order of \emph{every} pair of centers, and we detect the pair that is hardest to preserve; namely, two centers placed symmetrically on different sides of the sample mean $\smean$. 
They have identical cost on $Q$, 
yet on $P$ their cost ratio grows linearly with the deviation of $\smean$ from the true mean. 
Any nonzero deviation therefore turns a tie on $Q$ into a strict order on $P$, and forcing this deviation below $\epsilon$ demands the sample mean to be a $\sqrt\epsilon$ factor closer than in the weak coreset case. 
This quadratic gap in the tolerated deviation is exactly the $\epsilon^{-1}$ versus $\epsilon^{-2}$ separation, pinpointing an inherent difference between the two notions, weak and stable coresets.
The proof appears in \Cref{sec:lower_bound_appendix}.

\paragraph{Upper bounds for $k > 1$ (\Cref{thm:kmeans-stable,thm:kmeans-strong}).}
The difficulty in moving from $k=1$ to general $k$ is that the clean factorization of the RCDA error no longer holds. For a single cluster, every candidate center sees the same sample mean, so a single concentration bound controls all candidates simultaneously. For general $k$, however, each candidate solution (set of $k$ centers) induces its own Voronoi partition.
When the candidate centers vary (``move''), points near cluster boundaries may switch their assigned center, producing an additional error term that depends on the assignment itself. Consequently, unlike the $1$-mean case, there is no single empirical quantity whose concentration controls the entire family of candidate solutions.
The standard approach to handling such infinite families relies on VC-dimension or covering-number arguments~\citep{DBLP:conf/icml/HuangJL23}. We take a different route.

Fix an optimal solution $C^*$ and parametrize every candidate solution by its \emph{displacement} $v=(v_1,\dots,v_k)$ from the optimal centers, i.e.\ the candidate is $\{\mu_1+v_1,\dots,\mu_k+v_k\}$. In this parametrization, applying the K\"onig--Huygens identity separately within each optimal cluster of $P$ and of the sample $Q$ expresses both costs in closed form, and their difference splits into two terms.
The first is the \emph{matched-cost error}, in which every point is assigned to the center displaced from its own optimal center, regardless of which candidate center is actually nearest.
This term depends only on the empirical weights and centroids of the optimal clusters, and is controlled by the same per-cluster concentration as in the $1$-mean case.
The second is the \emph{boundary-gain error}, the correction for points whose nearest candidate center is displaced from a different optimal center.
Bounding this term is the core technical challenge.

Before turning to it, we exploit the geometry of well-separated instances to show that every candidate of constant-factor approximation is \emph{localized} around $C^*$: each of its centers lies within a small radius $\Gamma$ of a distinct optimal center, so the displacements lie in a fixed set $\Vset\subset\Xspace^k$. Crucially, the set $\Vset$ is determined by the input alone and does not depend on the sample. 

The boundary-gain error depends on the minimum squared distance from each point to the candidate centers.
This minimum is taken over all $k$ centers at once, making the error non-linear and non-separable across clusters.
To handle it, we first symmetrize the deviation, replacing the comparison between the sample and the population by a comparison between two independent samples, which yields a Rademacher process indexed by the displacements in $\Vset$.
The boundary gain at a point is a maximum of functions of $v$ that are affine except for a quadratic term, and we remove the two obstructions in turn.
Maurer's vector contraction inequality~\citep{DBLP:conf/alt/Maurer16} removes the maximum; it bounds a Rademacher process indexed by vector-valued Lipschitz functions, up to a constant factor, by a Rademacher process indexed by the coordinates of those functions.
The quadratic term is then removed by \emph{lifting} each displacement $v_r$ to $(v_r,\norm{v_r}^2/2\Gamma)$ in one extra dimension, which makes every coordinate function linear in the lifted variable.
The entire process thereby collapses to a single inner product in a Hilbert space between a bounded lifted displacement and one Rademacher vector, and Cauchy--Schwarz reduces the problem to the second moment of that vector, which we bound directly, without any tail assumption on the dataset $P$.
Overall, we obtain a uniform bound on the \emph{expected} boundary-gain error.

The final step is to convert this expectation bound into a high-probability guarantee, and this is exactly where our two results (\Cref{thm:kmeans-stable,thm:kmeans-strong}) diverge. For the strong coreset guarantee, we bound the variance of the boundary-gain error directly. Since changing one sampled point may have a large effect when that point lies far from its optimal cluster center, the resulting variance bound depends on fourth moments, which we control by the bounded-kurtosis assumption.
For the stable coreset guarantee, rather than approximating the cost itself, we first establish the $\epsilon$-RCDA guarantee. The RCDA decomposition eliminates the term responsible for the fourth-moment dependence before concentration is applied, leaving only second-moment quantities to control.
The trade-off is that this requires a bound on $\cost(C^*,Q)$, which follows easily from Markov's inequality, leading to the weaker dependence on the failure probability.
The full proof appears in \Cref{sec:kmean} due to space constraints.

\subsection{Related work}\label{sec:related}
A large literature studies algorithms for Euclidean $(k,z)$-clustering under structural assumptions on the input, including
perturbation resilience~\citep{DBLP:journals/cpc/BiluL12,BalcanLiang,AMM2017,BHW2016,DBLP:conf/innovations/MakarychevM20}, various
notions of separation and stability~\citep{ABS2010,ABS2012,rabani,CAS2017} and more~\citep{KumarKannan, BBG2013, 7606826}.

Another line of work studies sublinear-time clustering algorithms based on uniform random sampling~\citep{DBLP:conf/soda/MishraOP01, mop04,DBLP:journals/rsa/CzumajS07, DBLP:journals/ml/Ben-David07}. These works show that, under various settings and assumptions, approximate solutions for $k$-median and $k$-means can be obtained from samples whose size is independent of the input size.

Coresets have proven to be a valuable algorithmic tool for the $(k,z)$-clustering problem and have been studied extensively for over two decades~\citep{DBLP:conf/stoc/Har-PeledM04, DBLP:journals/siamcomp/Chen09,feldman2011unified, DBLP:conf/focs/SohlerW18,DBLP:conf/stoc/HuangV20, DBLP:conf/soda/BravermanJKW21, DBLP:conf/stoc/Cohen-AddadSS21}.
The dominant notion throughout this literature is that of a strong coreset, and the current state-of-the-art construction is essentially optimal, matching the lower bound up to lower-order
terms~\citep{DBLP:conf/stoc/Cohen-AddadLSS22,DBLP:conf/stoc/Huang0024}.
Strong coreset construction based on uniform sampling was originally proposed by~\citet{DBLP:journals/siamcomp/Chen09}, who suggested partitioning the metric space into rings and sample uniformly from each ring. The analysis was later improved in~\citep{DBLP:conf/nips/Cohen-AddadSS21,DBLP:conf/focs/BravermanCJKST022}.
Yet, the algorithm still requires a linear scan of the dataset and hence the overall sampling distribution is non-uniform.
For kernel clustering~\citet{DBLP:journals/ml/JiangKLZ24} give strong coresets of size $\poly(k/\epsilon)$, which as a special case yields strong coresets in $\ell_1^d$ (and every submetric of it).

Another related problem is Euclidean $1$-center (the minimum enclosing
ball), for which strong coresets require size exponential in the
dimension~\citep{AHV05}, whereas weak coresets are dimension-free~\citep{BHI02,BC08}. In the $\ell_1$ metric, however,
\emph{both} strong and weak coresets for $1$-center must depend on the dimension~\citep{DBLP:conf/innovations/CarmelGJK25}.

A related but conceptually different coreset notion, which is also intermediate between weak and strong coresets, is the \emph{approximation-preserving coreset} of~\citep{prabhu2026approximationpreservingcoresets}, which guarantees that an $\alpha$-approximate solution on $Q$ can be post-processed to an $\alpha\beta$-approximate solution on $P$, for every $\alpha,\beta\ge1$, via an explicit algorithm.

\section{Preliminaries}\label{section:preliminaries}

Let $\Xspace$ be a Hilbert space equipped with inner product $\ip{\cdot}{\cdot}$. Throughout the remainder of the paper, we focus on the $k$-means problem, corresponding to the special case $z=2$ of $(k,z)$-clustering. Accordingly, we omit the subscript $z$ from the notation.
In the \emph{$k$-means problem} the input is a set $P \subseteq \Xspace$
of size $|P| = n$ and an integer $k \geq 1$, 
and the goal is to find a $k$-center set $C \subseteq \Xspace$ minimizing
\[
  \cost(C, P) \eqdef \frac{1}{n} \sum_{p \in P} \min_{c \in C} \norm{p - c}^2.
\]
We denote by $\OPT \eqdef \min_{C \subseteq \Xspace, \abs{C}=k} \cost(C,P)$ the cost of an optimal solution.
Each center set $C$ induces a (Voronoi) partition of $P$ into
clusters $\{X_1,\ldots,X_k\}$, where $p \in X_i$ if $c_i$ is a nearest
center to $p$ in $C$ (ties broken by smallest label).
We fix an optimal solution $C^* = \{\mu_1,\ldots,\mu_k\}$ for the $k$-means problem with induced partition $\{X_1^*,\ldots,X_k^*\}$, for $k=1$ we simply denote $C^*=\mu \in \Xspace$. 
Define the per-cluster quantities 
\[
  w_i \eqdef \frac{|X_i^*|}{n} ,
  \quad
  V_i \eqdef \frac{1}{|X_i^*|}\sum_{p \in X_i^*}\norm{p - \mean_i}^2,
  \quad \text{ then }
  \cost(C^*, P) = \sum_{i=1}^k w_i V_i.
\]

The following identity decomposes the cost of a single cluster 
into a bias term and a variance term (also known as the K\"{o}nig--Huygens identity).

\begin{claim}\label{cl:kh}
Let $X \subseteq P$ of size $m$ and $\mu = \frac{1}{m}\sum_{p \in X} p$. 
Then, 
\begin{equation}
  \forall c \in \Xspace,
  \qquad
  \cost(c,X)
  =  \norm{c - \mean}^2 + \cost(\mu,X).
\end{equation}
\end{claim}

\begin{proof}
Write $p - c = (p - \mean) + (\mean - c)$. 
Then 
\[
  \cost(c,X)
  = \frac{1}{m} \sum_{p \in X} \norm{p - c}^2
  = \frac{1}{m}\sum_{p \in X}\norm{p-\mean}^2
    + 2\Bigl\langle\frac{1}{m}\sum_{p \in X}(p - \mean),\;\mean - c\Bigr\rangle
    + \norm{c - \mean}^2,
\]
and the middle term vanishes since $\frac{1}{m}\sum_{p \in X}(p - \mean) = 0$.
\end{proof}

It follows immediately that the centroid of $X_i^*$, 
i.e., $\mean_i = \frac{1}{|X_i^*|}\sum_{p \in X_i^*} p$,
is the point minimizing the cost for the cluster $X_i^*$. 

Let $Q$ be a uniform sample of size $s$ drawn i.i.d.\ from $P$.
Define the per-cluster sample quantities
\[
  Q_i \eqdef Q \cap X_i^*,
  \quad
  s_i \eqdef |Q_i|,
  \quad
  \hat{w}_i \eqdef \frac{s_i}{s},
  \quad
  \smean_i \eqdef \frac{1}{s_i}\sum_{q \in Q_i} q,
  \quad
  \Delta_i \eqdef \smean_i - \mean_i,
  \quad
  \hat V_i \eqdef \tfrac1{s_i}\sum_{q\in Q_i}\norm{q-\smean_i}^{2}.
\]

Now denote $\OPTQ \eqdef \sum_i \hat w_i \hat V_i$, and notice that this cost and its underlying partition are not necessarily optimal for $Q$, that is, 
$\min_{C \subseteq \Xspace, |C|=k} \cost(C,Q) \leq \OPTQ$.

A key tool for establishing stable coresets, introduced in~\citep{carmel2025stablecoresetsunleashingpower}, is the following notion, which requires that the cost gap between any candidate center set $C$ and the fixed optimal solution $C^*$ be nearly the same whether measured on $P$ or on the sample $Q$.

\begin{definition}[$\epsilon$-RCDA]\label{def:rcda}
A subset $Q \subseteq P$ is an
\emph{$\epsilon$-relative cost-difference approximation ($\epsilon$-RCDA)} of $P$ in $\Xspace$ if
{
\small
\begin{equation}\label{eq:rcda}
  \forall\, C \subseteq \Xspace,\ |C|=k,
  \qquad
  \Bigl|
    \bigl[\cost(C, P) - \cost(C^*, P)\bigr]
    -
    \bigl[\cost(C, Q) - \cost(C^*, Q)\bigr]
  \Bigr|
  \;\leq\; \epsilon  \cost(C, P).
\end{equation}
}
\end{definition}

The $\epsilon$-RCDA property implies a stable coreset, 
provided that the sample does not excessively overestimate the cost of the fixed optimum $C^*$.
The proof given in~\citep{carmel2025stablecoresetsunleashingpower}  for $1$-median
extends verbatim to $k$-means, for completeness we reproduce the proof in Appendix~\ref{appendix:rcda}.  

\begin{restatable}[\citet{carmel2025stablecoresetsunleashingpower} (Theorem~3.1)]{theorem}{rcdatostable}\label{thm:rcda-to-stable}
Let $\epsilon \in (0,1/5)$ and suppose that
\begin{enumerate}[label=(I\arabic*),nosep]
    \item \label{item:I1} $Q$ is an $\epsilon$-RCDA of $P$ in $\Xspace$, and
    \item \label{item:I2} $\cost(C^*, Q) \leq \alpha \cdot \cost(C^*, P)$ for some $\alpha \geq 1$.
\end{enumerate}
Then $Q$ is a stable $(\epsilon/\alpha,\,4\epsilon)$-coreset of $P$.
\end{restatable}

% ==================================================================
\section{Upper bound for \texorpdfstring{$k=1$}{k=1}}\label{sec:onecenter}

\subsection{RCDA in Hilbert space}

In the $k=1$ case, we denote the centers by lower-case letters and simply write $\cost(x,P) = \frac{1}{n}\sum\norm{x-p}^2$. Also denote $V=\cost(\mu,P)$, $V_Q =\cost(\smean, Q)$, where $\smean\eqdef\frac1s\sum_{q\in Q}q$ is the sample mean of $Q$ and $\Delta\eqdef\smean-\mean$.

Applying \cref{cl:kh} to both $P$ and $Q$, for all $c \in \Xspace$:
\begin{align}
  \cost(c, P) &= \norm{c - \mean}^2 + V, \qquad&
  \cost(\mean, P) &= V, \label{eq:kh-P}\\
  \cost(c, Q) &= \norm{c - \smean}^2 + V_Q, \qquad&
  \cost(\mean, Q) &= \norm{\Delta}^2 + V_Q. \label{eq:kh-Q}
\end{align}

First, we note that the RCDA error can be written as a single inner product.

\begin{claim}\label{cl:rcda-ip}
For every $c \in \Xspace$,
\[
  \bigl[\cost(c,P) - \cost(\mean,P)\bigr] - \bigl[\cost(c,Q) - \cost(\mean,Q)\bigr]
  \;=\; 2\ip{c - \mean}{\Delta}.
\]
\end{claim}
\begin{proof}
From \eqref{eq:kh-P} and \eqref{eq:kh-Q}:
\begin{align*}
  \cost(c,P) - \cost(\mean,P) &= \norm{c-\mean}^2 + V - V = \norm{c-\mean}^2,\\
  \cost(c,Q) - \cost(\mean,Q) &= \bigl(\norm{c-\smean}^2 + V_Q\bigr) - \bigl(\norm{\Delta}^2 + V_Q\bigr)
                                = \norm{c-\smean}^2 - \norm{\Delta}^2.
\end{align*}
Expanding $\norm{c - \smean}^2 = \norm{(c-\mean) - \Delta}^2 = \norm{c-\mean}^2 - 2\ip{c-\mean}{\Delta} + \norm{\Delta}^2$, so 
\begin{align*}
    \cost(c,Q)-\cost(\mean,Q) = \norm{c-\mu}^2-2\ip{c-\mu}{\Delta}.
\end{align*}
Consequently,
\begin{align*}
    \bigl[\cost(c,P) - \cost(\mean,P)\bigr] - \bigl[\cost(c,Q) - \cost(\mean,Q)\bigr]
  = \norm{c-\mu}^2 - (\norm{c-\mu}^2-2\ip{c-\mu}{\Delta}) =   
   2\ip{c - \mean}{\Delta}.
\end{align*}
\end{proof}

Second, we bound $\norm{\Delta}^2$ via concentration (see also~\citep{DBLP:conf/compgeom/InabaKI94} for similar concentration analysis).

\begin{claim}\label{cl:concentration_delta}
Let $\epsilon>0$, then $\Pr\bigl[\norm{\Delta}^2 > \epsilon^2 V\bigr] \leq \dfrac{1}{s\epsilon^2}$.
\end{claim}

\begin{proof}
We can write $\Delta = \frac{1}{s}\sum_{q \in Q}(q - \mean)$, where the summands are i.i.d.\
with mean zero.
By independence,
\[
  \mathbb{E}\bigl[\norm{\Delta}^2\bigr]
  = \frac{1}{s^2}\sum_{q \in Q}\mathbb{E}\bigl[\norm{q - \mean}^2\bigr]
  = \frac{1}{s^2} \cdot s \cdot V = \frac{V}{s}.
\]
Applying Markov's inequality 
\[
  \Pr\bigl[\norm{\Delta}^2 > \epsilon^2 V\bigr]
  \;\leq\; \frac{\mathbb{E}[\norm{\Delta}^2]}{\epsilon^2 V}
  = \frac{1}{s\epsilon^2}. \qedhere
\]
\end{proof}

We combine Claim~\ref{cl:rcda-ip} and Claim~\ref{cl:concentration_delta} to show that a uniformly sampled $Q$ is an $\epsilon$-RCDA of $P$.

\begin{claim}\label{cl:rcda-bound}
If $\norm{\Delta}^2 \leq \epsilon^2 V$, then $Q$ is an $\epsilon$-RCDA of~$P$.
\end{claim}
\begin{proof}
By \cref{cl:rcda-ip}, the RCDA error at center $c$ equals $|2\ip{c-\mean}{\Delta}|$.
By Cauchy--Schwarz and Young's inequality ($2ab \leq \epsilon a^2 + b^2/\epsilon$):
\[
  |2\ip{c-\mean}{\Delta}|
  \;\leq\; 2\norm{c-\mean}\cdot\norm{\Delta}
  \;\leq\; \epsilon\norm{c-\mean}^2 + \frac{1}{\epsilon}\norm{\Delta}^2
  \;=\; \epsilon (\cost(c,P) - V)+ \frac{1}{\epsilon}\norm{\Delta}^2.
\]
By assumption, $\frac{1}{\epsilon}\norm{\Delta}^2 \leq \epsilon V$.
Hence, $|2\ip{c-\mean}{\Delta}| \leq \epsilon\cdot\cost(c,P)$, for all $c \in \Xspace$.
\end{proof}

\subsection{Proof of Theorem~\ref{thm:stable}}

\begin{proof}[Proof of \cref{thm:stable}]
Let $\epsilon \leq 1/5$, $\delta >0$ and set $s \eqdef \lceil 2/(\delta \epsilon^2) \rceil$.
By \cref{cl:concentration_delta}, $\Pr[\norm{\Delta}^2 > \epsilon^2 V] \leq \delta/2$.
Since $\mathbb{E}[\cost(\mean, Q)] = V$, Markov's inequality gives
$\Pr[\cost(\mean, Q) > 2V/\delta] \leq \delta/2$.
By the union bound, both $\norm{\Delta}^2 \leq \epsilon^2 V$ and
$\cost(\mean, Q) \leq 2V/\delta$ hold simultaneously with probability at least $1-\delta$.
Conditioned on these events, \cref{cl:rcda-bound} gives that $Q$ is an
$\epsilon$-RCDA of $P$, and $\cost(\mean, Q) \leq 2V/\delta = 2\cost(\mean, P)/\delta$.
Applying \cref{thm:rcda-to-stable},
$Q$ is a stable $(\delta\epsilon/2,\, 4\epsilon)$-coreset of $P$, and the desired result is obtained by rescaling $\epsilon$.
\end{proof}

\subsection{Proof of Theorem~\ref{thm:weak}}

In contrast, we show that for weak coresets it suffices to
have $\norm{\Delta}^2 \leq \epsilon V$, a single power of $\epsilon$ rather than
$\epsilon^2$ as in \cref{cl:rcda-bound}, thus improving upon the required sample size.
While a slightly weaker result, yielding a weak $(0,\epsilon)$-coreset via a similar technique, appears in~\citep{DBLP:conf/compgeom/InabaKI94, DBLP:conf/icml/JubranTMF20}, we include in Appendix~\ref{appendix:weak} the proof of the more general statement for completeness.

\weakthm*

\section{Experiments}\label{section:experiments}

We evaluate uniform sampling against importance-sampling coreset constructions on
widely-used real-world datasets (summarized in Table~\ref{table:datasets}).
Generally speaking, the accuracy of uniform sampling in our experiments is consistent 
with experiments in prior work~\citep{DBLP:conf/nips/MaromF19,
DBLP:conf/icml/BakerBHJK020, DBLP:conf/icml/JubranTMF20, 
DBLP:conf/nips/BravermanJKW21,
DBLP:conf/icml/HuangJL23}. 
These standard datasets were not chosen to satisfy the requirements of our analysis,
and indeed they do not fully satisfy them.
At every value of $k$, the measured separation is $\hat\rho^2 \ge 0.866$, 
while the measured kurtosis is less than $5$. 

Our results concern the three coreset guarantees: weak, stable and strong. 
For each, we define an empirical error that instantiates the corresponding definition, 
and we report these errors as percentages.
The empirical weak-coreset error instantiates \eqref{eq:weak} with $\epsilon=0$, as 
\begin{equation}\label{eq:weak_error}
  \widehat{E}_{\mathrm{weak}}
  := \frac{\cost(C^Q,P)-\cost(C^P,P)}{\cost(C^P,P)},
\end{equation}
where $C^Q$ and $C^P$ are optimal solutions for $Q$ and for
$P$. 
The empirical stable-coreset error instantiates \eqref{eq:stable}, as
\begin{equation}\label{eq:stable_error}
  \widehat{E}_{\mathrm{stable}}
  := \max\Bigl\{
      \tfrac{\cost(C_1,P)}{\cost(C_2,P)}-1
      \ :\ C_1,C_2\in\mathcal{C},\
      \cost(C_1,Q)\le\cost(C_2,Q)
    \Bigr\}. 
\end{equation}
The empirical strong-coreset error instantiates \eqref{eq:strong}, as
\begin{equation}\label{eq:strong_error}
  \widehat{E}_{\mathrm{strong}}
  := \max_{C\in\mathcal{C}}
    \Bigl|\tfrac{\cost(C,Q)}{\cost(C,P)}-1\Bigr| .
\end{equation}

Whereas Definitions~\eqref{eq:stable} and~\eqref{eq:strong} quantify over every
center set in $\Xspace$, in \eqref{eq:stable_error} and
\eqref{eq:strong_error} the maximum is taken over a finite family
$\mathcal{C}$ of candidate solutions, that we generated once per dataset and
reused across every method and coreset size; both are therefore lower
bounds on the true error. 
In Experiment~2, $\mathcal{C}$ consists of
$5{,}000$ sets of $k$ points drawn uniformly from $P$. 
In Experiment~1,
$\mathcal{C}$ consists of centers along the direction that the proof of \Cref{thm:lowerbounds}\ref{item:stable} identifies as hardest.

\textbf{Experimental setup.} All experiments were conducted on a PC with Intel Core i7-13700 CPU and 32GB RAM running Python 3.12.10 on Windows 10. For each experiment, we report the average results over 20 independent runs with shaded regions representing one standard deviation across the runs. 

\begin{table}[t]
\centering
\footnotesize
\setlength{\tabcolsep}{2pt}
\caption{\footnotesize{Specifications of datasets used in Section~\ref{section:experiments}.}}
\label{table:datasets}
\begin{tabular}{>{\raggedright\arraybackslash}p{4.5cm}rr>{\raggedright\arraybackslash}p{6.5cm}}
\toprule
\textbf{Dataset} & \textbf{Size $n$} & \textbf{Dim.\ $d$} & \textbf{Description} \\
\midrule
Yellow Taxi NYC (YT)~\citep{nyctlc2024} & 2.82M & 16 & New York City taxi trips in Jan.\ 2024 \\
Covertype~\citep{covertype_31} & 581,012 & 54 & Cartographic variables for forest plots in the Roosevelt National Forest of Northern Colorado \\
Census 1990~\citep{Census1990} & 2.45M & 68 & Demographic attributes of individuals from the US Census Public Use Microdata Samples \\
\bottomrule
\end{tabular}
\end{table}

\paragraph{Experiment 1: Comparison with importance sampling.}
We compare uniform sampling against sensitivity
sampling~\citep{feldman2011unified} on the $1$-median problem in the
$\ell_1$ metric over the Yellow Taxi NYC dataset. 
We vary the coreset size over $\{200, 250, \ldots, 1000\}$ 
and report both the empirical weak-coreset error~\eqref{eq:weak_error} and the empirical stable-coreset error~\eqref{eq:stable_error}. 
We find that the stable-coreset error exceeds the weak-coreset error roughly by an order of magnitude at every coreset size, reflecting that preserving the order of costs is a substantially stronger requirement than preserving near-optimality alone.

Figure~\ref{figure:comparison} shows no appreciable difference between
the two sampling methods, under either notion.
The difference in construction time is, in contrast, fundamental:
sensitivity sampling examines the entire dataset before it can draw a single point, 
whereas uniform sampling requires constant time per sample and is oblivious to the dataset. 
And indeed, producing a coreset of $500$
points took $2{,}160$ milliseconds for sensitivity sampling against $0.03$ milliseconds for uniform sampling —
almost five orders of magnitude faster.

\begin{figure}[t]
    \centering
    \includegraphics[width=0.8\linewidth]{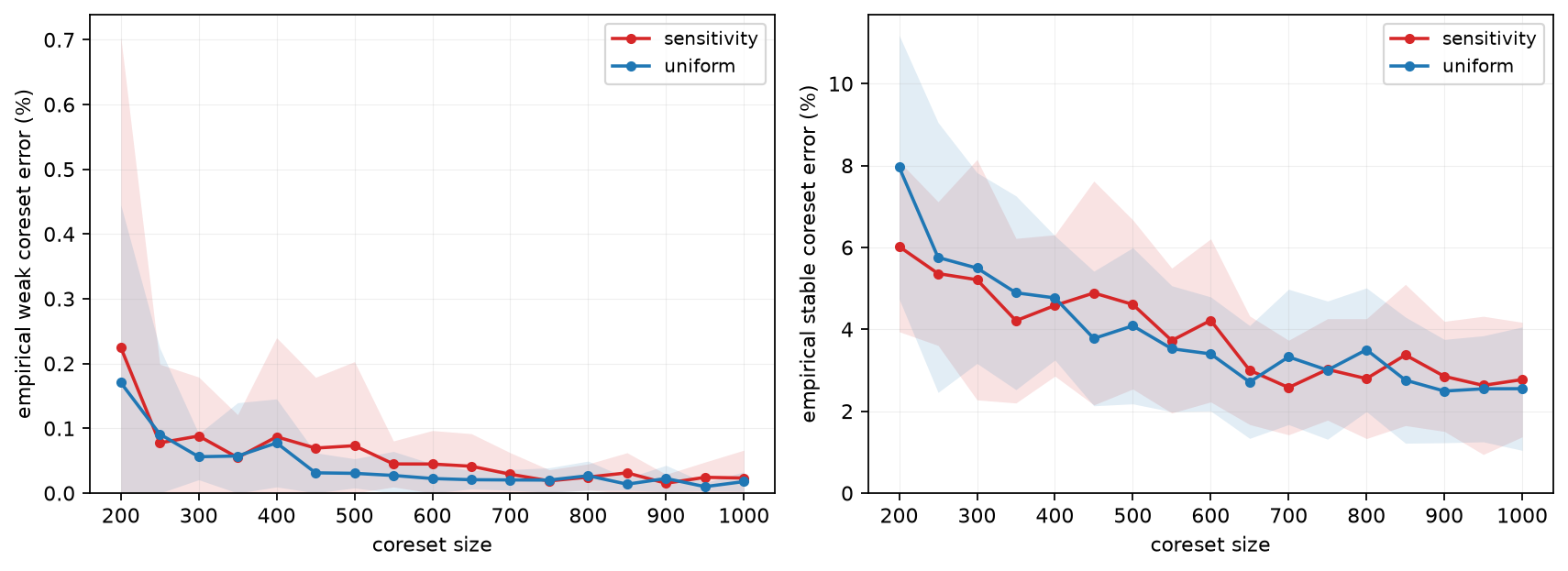}
    \caption{Comparison of uniform and sensitivity sampling on the Yellow Taxi NYC
    dataset: empirical weak-coreset error (left) and empirical stable-coreset error
    (right) as a function of coreset size, for $1$-median in $\ell_1$.}
    \label{figure:comparison}
\end{figure}

\paragraph{Experiment 2: Strong coresets for $k$-means.}
We compare uniform sampling against two importance-sampling
constructions, Lightweight Coresets~\citep{lightweight} and Sensitivity
Sampling~\citep{feldman2011unified}, on the $k$-means problem over the
Census 1990 and Covertype datasets, and report the empirical strong-coreset error~\eqref{eq:strong_error}. We vary the coreset size over $\{200, 250, \ldots, 1000\}$. 
Each dataset is run at two values of $k$: Covertype at its seven cover types and at $35$,
and Census, which carries no label, at $10$ and $50$. 

Figure~\ref{figure:kmeans} shows that uniform sampling is
comparable to the two importance-sampling methods, 
even though the assumptions of \Cref{thm:kmeans-strong} are not satisfied. 
As in Experiment 1, there is a dramatic difference in construction time.
On the larger value of $k$, producing a coreset of $1{,}000$ points took $20{,}000$ and $127{,}000$ milliseconds for sensitivity sampling on Covertype and Census respectively, and $90$ and $900$ milliseconds for lightweight coresets, against $0.05$ milliseconds for uniform sampling —
at least three orders of magnitude faster and sometimes much more. 

\begin{figure}[t]
    \centering
    \includegraphics[width=0.8\linewidth]{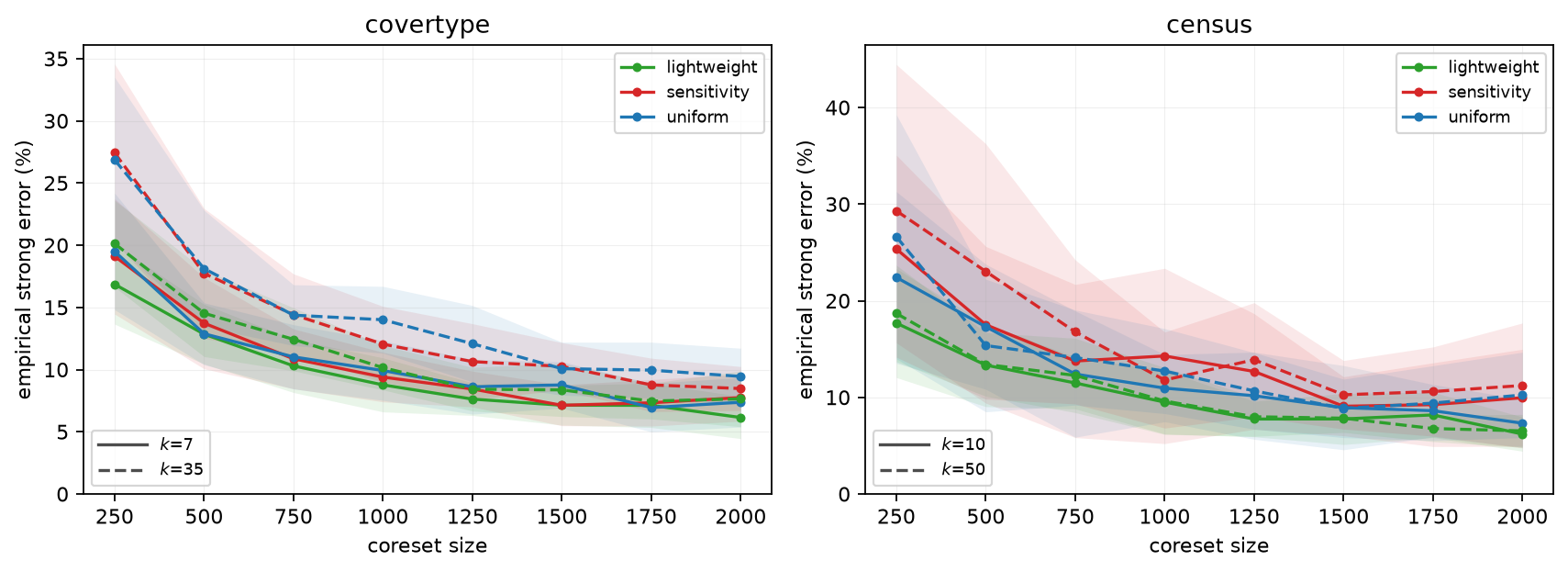}
    \caption{Empirical strong coreset error as a function of coreset
    size, on Covertype (left) and Census 1990 (right). Color denotes the
    sampling method and line style the value of $k$.}
    \label{figure:kmeans}
\end{figure}

\bibliographystyle{plainnat} 
\bibliography{bibliography}

\appendix

\section{Proof of Theorem~\ref{thm:rcda-to-stable}}\label{appendix:rcda}

\rcdatostable*

\begin{proof}
For every $C \subseteq \Xspace$ with $|C|=k$, using~\ref{item:I1} and~\ref{item:I2},
\begin{equation}\label{eq:rcda-cost-bound}
\cost(C,Q) \leq (\alpha-1)\cost(C^*,P) + (1+\epsilon)\cost(C,P) \leq (\alpha+\epsilon)\cost(C,P).
\end{equation}
For every $C_1,C_2 \subseteq \Xspace$ of size $k$,
we have by the triangle inequality and property~\ref{item:I1}, 
\[
 \Bigl|
    \bigl[\cost(C_1, P) - \cost(C_2, P)\bigr]
    -
    \bigl[\cost(C_1, Q) - \cost(C_2, Q)\bigr]
  \Bigr|
  \;\leq\; \epsilon  (\cost(C_1, P)+\cost(C_2, P)).
\]
Now if $\cost(C_1, Q) \leq (1+\epsilon/\alpha)\cost(C_2, Q)$, then
\begin{align*}
\cost(C_1,P) &\leq \cost(C_2,P) + (\cost(C_1, Q) - \cost(C_2, Q))+ \epsilon  (\cost(C_1, P)+\cost(C_2, P)) \\
&\leq \cost(C_2,P) + \frac{\epsilon}{\alpha}\cost(C_2, Q)+ \epsilon  (\cost(C_1, P)+\cost(C_2, P)) \\
&\leq \cost(C_2,P) + \frac{\epsilon}{\alpha}(\alpha+\epsilon)\cost(C_2,P)+ \epsilon  (\cost(C_1, P)+\cost(C_2, P)) \\
&= (1+2\epsilon+\tfrac{\epsilon^2}{\alpha})\cost(C_2,P)+\epsilon\cost(C_1,P),
\end{align*}
where the third line applies~\eqref{eq:rcda-cost-bound} to $C_2$.
By rearranging,
$\cost(C_1,P) \leq \frac{1}{1-\epsilon}(1+2\epsilon+\frac{\epsilon^2}{\alpha})\cost(C_2,P) \leq (1+4\epsilon)\cost(C_2,P)$,
where the last inequality uses $\epsilon \leq 1/5$.
\end{proof}

\section{Proof of Theorem~\ref{thm:weak}}\label{appendix:weak}

\weakthm* 

\begin{proof}
Let $c \in \Xspace$ and let $Q$ be a uniform sample of size $s \geq 2/(\epsilon \delta)$.
If $\cost(c,Q) \le (1+\delta\epsilon/2)V_Q$, then by \eqref{eq:kh-Q},
$\norm{c-\smean}^2 \leq \delta \epsilon V_Q/2$,
and we need to show that $\norm{c-\mean}^2 \leq 4\epsilon V$.
Observe
\begin{align*}
    \norm{c-\mean}^2 
    = \norm{c-\smean+\smean-\mean}^2 
    \leq 2\norm{c-\smean}^2 + 2\norm{\smean-\mean}^2 
    \leq \delta \epsilon V_Q + 2\norm{\Delta}^2.
\end{align*}

It remains to bound $V_Q$ and $\norm{\Delta}^2$.
As in Claim~\ref{cl:concentration_delta}, 
$ \Pr\bigl[\norm{\Delta}^2 > \epsilon V\bigr] \leq \frac{1}{s\epsilon} \leq \frac{\delta}{2} .$
Using Claim~\ref{cl:kh}, $\cost(\mean,Q) = \norm{\Delta}^2 + V_Q$,
and therefore $V_Q \leq \cost(\mean,Q)$. 
Since $\mathbb{E}[\cost(\mean,Q)] = V$, Markov's inequality gives
$
\Pr\bigl[V_Q > \tfrac{2}{\delta} V\bigr]
\leq
\Pr\bigl[\cost(\mean,Q) > \tfrac{2}{\delta} V\bigr]
\leq \frac{\delta}{2}.
$
Conditioned on both events, $
\norm{c-\mean}^2
\leq
\delta\epsilon \cdot \frac{2}{\delta}V + 2 \cdot \epsilon V
=
4\epsilon V$, with probability $1-\delta$. The result is obtained by rescaling $\epsilon$.
\end{proof}

% ==================================================================
\section{Lower bounds for \texorpdfstring{$k=1$}{k=1}}\label{sec:lower_bound_appendix}

The proofs of the lower bounds use a two-point instance, mass $1/n$ at $D\neq 0$ and the rest at $0$, with $n$ tuned to $s$ differently in each proof. We prove the weak bound first, as the stable construction reuses the same instance; its proof also appears in~\citep{DBLP:conf/nips/Cohen-AddadSS21}, and we include it to contrast with the stable guarantee.

\begin{proof}[Proof of \cref{thm:lowerbounds}~\ref{item:weak}]
    Let $n =  \frac{10}{3}s( \epsilon s +1) $, Let $P$ be such that $1/n$ points are at $D \neq 0$ and the rest are at $0$ (w.l.o.g. $n$ divides the size of $P$). Let $X$ be a random variable denoting the number of sampled points at $D$ in $Q$, then, $X \sim \mathrm{Bin}(s,p)$ with $p=\frac{1}{n}$.
    We have,
    \begin{align}\label{eq:weaklb}
        \mean = pD, \quad V=p(1-p)D^2, \quad \Delta = D(\frac{X}{s}-p).
    \end{align}
    If $Q$ is a weak $(0,\epsilon)$-coreset then $\cost(\smean,P) \leq (1+\epsilon)\OPT$, that is, $V + \norm{\Delta}^2 \leq (1+\epsilon)V$.
    By \eqref{eq:weaklb} it follows that $\abs{\frac{X}{s}-p} \leq \sqrt{\epsilon p (1-p)}$.
    Conditioned on the event $X=1$, since $n\geq 10s/3$ we have $\abs{\frac{X}{s}-p} \geq \frac{7}{10s}$, while $\sqrt{\epsilon p (1-p)} \leq \sqrt{\epsilon p}\leq \sqrt{\frac{3\epsilon}{10s(\epsilon s+1)}} \leq \frac{\sqrt{3}}{\sqrt{10}s}<\frac{3}{5s}$. Consequently, conditioned on the event $X=1$, $Q$ is not a weak $(0,\epsilon)$-coreset. As,
    \[
    \delta \geq \Pr\bigl[X=1\bigr] = sp(1-p)^{s-1}\geq sp(1-p)^{n}\geq \frac{1}{3}sp \geq \frac{1}{10\epsilon s+10}.
    \]
    
    We get $10\epsilon s \delta +10\delta \geq 1$, since $\delta \leq 1/20$, it follows that $s \geq \frac{1}{20\epsilon\delta}$.
\end{proof}

The stable lower bound instead controls the cost ratio between two centers symmetric about $\smean$, rather than the displacement $\Delta$ directly.

\begin{proof}[Proof of \cref{thm:lowerbounds}~\ref{item:stable}]

Let $P$ and $X$ be as in the proof of \cref{item:weak} and further denote $t=\sqrt{V+\Delta^2}$.

    Let $c_1=\smean+\mathrm{sign}(\Delta)\,t$ and $c_2=\smean-\mathrm{sign}(\Delta)\,t$. Since $c_1,c_2$ are symmetric around $\smean$, Claim~\ref{cl:kh} gives $\cost(c_1,Q)=\cost(c_2,Q)$. Using $\cost(c,P)=V+(c-\mean)^2$, we get

    \[\frac{\cost(c_1,P)}{\cost(c_2,P)} = \frac{V+(\Delta+t)^2}{V+(\Delta-t)^2}= \frac{2V+2\Delta^2+2\Delta\sqrt{V+\Delta^2}}{2V+2\Delta^2-2\Delta\sqrt{V+\Delta^2}} = \frac{\sqrt{V+\Delta^2}+|\Delta|}{\sqrt{V+\Delta^2}-|\Delta|}.\] 
    
    By setting $x \eqdef \frac{|\Delta|}{\sqrt{V}}$, this ratio equals $\frac{\sqrt{1+x^2}+x}{\sqrt{1+x^2}-x}$, which implies $\frac{\cost(c_1,P)}{\cost(c_2,P)}-1 \ge 2x$.

    \paragraph{Case 1: $s< 4/\epsilon^2$. } Set $n=s$. Conditioned on $X=0$,
    \begin{align*}
        x = \frac{\abs{\Delta}}{\sqrt{V}} = \frac{\abs{(\frac{X}{s} - p)D}}{\sqrt{p(1-p)} \abs{D}} \geq \frac{\frac{1}{n}}{\sqrt{\frac{1}{n}}} = \frac{1}{\sqrt{n}} >\frac{\epsilon}{2}.
    \end{align*}
    It follows that $\frac{\cost(c_1,P)}{\cost(c_2,P)} > 1+\epsilon$, and therefore, conditioned on $X=0$, $Q$ is not a stable $(0,\epsilon)$-coreset.

    As,
    $$
    \delta \ge \Pr\bigl[X=0\bigr] = (1-p)^{s}  = (1-\frac{1}{n})^n \geq \frac{1}{5}
    $$
    Therefore the sample cannot succeed with probability at least $1-\delta > 4/5 $.

    \paragraph{Case 2: $s \geq 4/\epsilon^2$. } Set $n=s^2\epsilon^2$. Conditioned on $X=1$, 
    \begin{align*}
        x = \frac{\abs{\Delta}}{\sqrt{V}} = \frac{\abs{(\frac{X}{s} - p)D}}{\sqrt{p(1-p)} \abs{D}} \geq \frac{\frac{1}{s}-\frac{1}{n}}{\sqrt{\frac{1}{n}}} = \frac{\sqrt{n}}{s} - \frac{1}{\sqrt{n}} = \epsilon - \frac{1}{s\epsilon} >\frac{\epsilon}{2}.
    \end{align*}
    It follows that $\frac{\cost(c_1,P)}{\cost(c_2,P)} > 1+\epsilon$, and therefore, conditioned on $X=1$, $Q$ is not a stable $(0,\epsilon)$-coreset.

    As,
    $
    \delta \ge \Pr\bigl[X=1\bigr] = sp(1-p)^{s-1} \geq \frac{1}{3s\epsilon^2}.
    $
    We get, $s \geq \frac{1}{3\epsilon^2\delta}$.
\end{proof}

\section{Upper bounds for \texorpdfstring{$k>1$}{k>1}}\label{sec:kmean}

In this section we establish our coreset guarantees for $k$-means. We begin in \cref{sec:stage} by introducing the structural assumptions of balancedness and separation, and derive several geometric consequences that are used throughout the analysis.

In \cref{sec:matching}, we parametrize candidate solutions by their \emph{displacement} from the optimal centers. This yields a master identity that splits the difference between the costs on $P$ and on $Q$ into a matched-cost error and a boundary-gain error, reducing the analysis to bounding these two terms separately.

In \cref{sec:event} we define a high-probability event $\event$ on which the sampled cluster weights and sample-mean deviations concentrate, and show that on this event every candidate solution of interest is \emph{localized}: its displacement lies in a fixed set $\Vset$, determined by the input alone. Localization is used only through this deterministic inclusion; all subsequent probabilistic statements are made uniformly over $\Vset$ and require no conditioning.

We then analyze the two error terms in \cref{sec:matched,sec:boundary}. We first bound the matched-cost error, and then establish a uniform bound on the boundary-gain error, which forms the core technical component of the analysis.

Finally, the analysis branches according to the desired guarantee. In \cref{sec:stable}, we combine the matched-cost and boundary-gain bounds with a one-sided estimate on the sampled optimal cost to obtain an RCDA guarantee, which is then converted into a stable coreset. In \cref{sec:strong}, we instead assume bounded kurtosis to strengthen the boundary-gain bound, leading to the strong coreset guarantee.

\subsection{Structural assumptions}\label{sec:stage}

We start by defining the balancedness property, which is inherent to the uniform
sampling approach.

\begin{definition}\label{def:balanced-data-k}
The dataset $P$ is \emph{$\beta$-balanced} for the $k$-means problem if there exists
an optimal $k$-center set $C^*=\{\mu_1,\dots,\mu_k\}$ whose induced partition
$\{X_1^*,\ldots,X_k^*\}$ satisfies $|X_i^*| \geq \beta n/k$ for every $i \in [k]$.
\end{definition}

Fix an optimal $C^*=\{\mu_1,\dots,\mu_k\}$ with induced partition $\{X_1^*,\ldots,X_k^*\}$ and set $w_i\eqdef|X_i^*|/n\ge\beta/k$,
$V_i\eqdef\frac1{|X_i^*|}\sum_{p\in X_i^*}\norm{p-\mu_i}^2$, and $\OPT= \OPT (P) = \sum_i w_iV_i$.
Since $\mu_i$ is the centroid of $X_i^*$, Claim~\ref{cl:kh} gives, for any $c$,
\begin{equation}\label{eq:kh-k}
  \tfrac1{|X_i^*|}\textstyle\sum_{p\in X_i^*}\norm{p-c}^2=\norm{c-\mu_i}^2+V_i .
\end{equation}
For $p\in P$, let $i(p)$ denote the index with $p\in X_{i(p)}^{*}$ and
$a(p)=\norm{p-\mu_{i(p)}}$. For convenience $\dist(p,C)=\min_{1\le j\le k}\norm{p-c_j}$.

The separation condition, which originated in~\citep{rabani}, requires that
the optimal cost drop substantially when passing from $k-1$ to $k$ centers.

\begin{definition}[$\rho$-separated,~\citep{rabani}]\label{def:orss}
    $P \subseteq \Xspace$ is $\rho$-separated for $k$-means if $\OPT_k \leq \rho^2 \OPT_{k-1}$,
    where $\OPT_k =\OPT$ and $\OPT_{k-1}$ is the optimal $(k-1)$-means cost.
\end{definition}

The following lemma allows us to bound the distance between
optimal centers via the $\rho$-separated guarantee (similar ideas appeared in~\citep{rabani}).

\begin{lemma}\label{lemma:rabani}
If $\rho^2\le1/2$ then $D\eqdef\min_{i\ne j}\norm{\mu_i-\mu_j}$ satisfies
$D^2\ge\OPT/\rho^2$.
\end{lemma}
\begin{proof}
Merging $X_i^*,X_j^*$ and keeping the other clusters changes the normalized cost by
$\frac{w_iw_j}{w_i+w_j}\norm{\mu_i-\mu_j}^2$, so by Claim~\ref{cl:kh} applied to the
two optimal centers,
$\OPT_{k-1}\le\OPT+\frac{w_iw_j}{w_i+w_j}\norm{\mu_i-\mu_j}^2$. Since $\OPT_{k-1}\ge\OPT/\rho^2$,
$\OPT(\rho^{-2}-1)\le\frac{w_iw_j}{w_i+w_j}\norm{\mu_i-\mu_j}^2\le\frac12\norm{\mu_i-\mu_j}^2$.
With $\rho^2\le1/2$, $2(\rho^{-2}-1)\ge\rho^{-2}$, giving $\norm{\mu_i-\mu_j}^2\ge\OPT/\rho^2$.
\end{proof}

Let $F=\{p\in P:a(p)\ge D/4\}$ denote the set of \emph{far} points and
$T=\frac1n\sum_{p\in P}a(p)^2\indic{p\in F}$ their cost. Let $Q$ be a uniform
i.i.d.\ sample of size $s$, drawn with replacement and regarded as a multiset;
set $Q_i\eqdef Q\cap X_i^*$, $s_i\eqdef|Q_i|$,
$\hat w_i\eqdef s_i/s$, $\smean_i\eqdef\frac1{s_i}\sum_{q\in Q_i}q$,
$\Delta_i\eqdef\smean_i-\mu_i$,
$\hat V_i\eqdef\frac1{s_i}\sum_{q\in Q_i}\norm{q-\smean_i}^2$, and
$\OPTQ\eqdef\sum_i\hat w_i\hat V_i$. Applying \eqref{eq:kh-k} within each $Q_i$,
\begin{equation}\label{eq:optq-split}
  \cost(C^*,Q)=\tfrac1s\textstyle\sum_{q\in Q}a(q)^2
  =\OPTQ+\sum_i\hat w_i\norm{\Delta_i}^2\ \ge\ \OPTQ .
\end{equation}

Note that $\OPTQ$ is the cost on $Q$ of the fixed partition $\{Q_i\}$ when each $Q_i$
is represented by its empirical centroid $\smean_i$; hence it upper-bounds the
optimal sample cost, $\OPT(Q)=\min_{|C|=k}\cost(C,Q)\le\OPTQ$.

\subsection{Displacements and the master identity}\label{sec:matching}

Our analysis parametrizes candidate solutions by their offset from the optimal
centers. For a \emph{displacement} $v=(v_1,\dots,v_k)\in\Xspace^k$ write
$C(v)\eqdef\{\mu_1+v_1,\dots,\mu_k+v_k\}$, and define the pointwise
\emph{boundary gain}
\begin{equation}\label{eq:bv}
  b_v(p)\eqdef\norm{p-\mu_{i(p)}-v_{i(p)}}^2-\min_{1\le j\le k}\norm{p-\mu_j-v_j}^2\ \ge 0 ,
\end{equation}
which measures how much a point gains by being served by its nearest center in
$C(v)$ rather than by the center displaced from its own optimal center. Set
$B_v(P)\eqdef\frac1n\sum_{p\in P}b_v(p)$ and $B_v(Q)\eqdef\frac1s\sum_{q\in Q}b_v(q)$.
The \emph{matched costs} are obtained by assigning every point to the center
displaced from its optimal center,
\[
  \costm(C(v),P)\eqdef\tfrac1n\textstyle\sum_p\norm{p-\mu_{i(p)}-v_{i(p)}}^2 ,
  \qquad
  \costm(C(v),Q)\eqdef\tfrac1s\textstyle\sum_q\norm{q-\mu_{i(q)}-v_{i(q)}}^2 ,
\]
so that $B_v(P)=\costm(C(v),P)-\cost(C(v),P)$ and
$B_v(Q)=\costm(C(v),Q)-\cost(C(v),Q)$.

Applying \eqref{eq:kh-k} within each $X_i^*$, and likewise within each $Q_i$ (whose
centroid is $\smean_i=\mu_i+\Delta_i$), the matched costs evaluate in closed form,
and we obtain
\begin{equation}\label{eq:excess}
\begin{aligned}
  \cost(C(v),P)&=\OPT+\textstyle\sum_i w_i\norm{v_i}^2-B_v(P),\\
  \cost(C(v),Q)&=\OPTQ+\textstyle\sum_i \hat w_i\norm{v_i-\Delta_i}^2-B_v(Q).
\end{aligned}
\end{equation}
Subtracting the two lines of \eqref{eq:excess} and expanding
$\norm{v_i-\Delta_i}^2$ gives the master identity: for every $v\in\Xspace^k$,
\begin{equation}\label{eq:master}
  \cost(C(v),P)-\cost(C(v),Q)
  =\bigl[\costm(C(v),P)-\costm(C(v),Q)\bigr]+\bigl[B_v(Q)-B_v(P)\bigr],
\end{equation}
\begin{equation}\label{eq:matched-expand}
\begin{aligned}
  & \costm(C(v),P)-\costm(C(v),Q)
  = \\
  & (\OPT-\OPTQ)+\sum_i(w_i-\hat w_i)\norm{v_i}^2
   +2\sum_i\hat w_i\ip{v_i}{\Delta_i}-\sum_i\hat w_i\norm{\Delta_i}^2 .
\end{aligned}
\end{equation}
The first bracket in \eqref{eq:master} is the \emph{matched-cost error} and the
second is the \emph{boundary-gain error}, arising from points whose nearest center
changes under sampling.

Using \eqref{eq:optq-split} to write
$\cost(C^*,Q)-\OPTQ=\sum_i\hat w_i\norm{\Delta_i}^2$, the last term of
\eqref{eq:matched-expand} cancels against the $\OPTQ$ term, and the RCDA error
of $C(v)$ becomes
\begin{equation}\label{eq:rcda-error}
\begin{aligned}
  E(v)&\eqdef\bigl(\cost(C(v),P)-\OPT\bigr)-\bigl(\cost(C(v),Q)-\cost(C^*,Q)\bigr)\\
  &=\sum_i(w_i-\hat w_i)\norm{v_i}^2+2\sum_i\hat w_i\ip{v_i}{\Delta_i}
   +\bigl[B_v(Q)-B_v(P)\bigr].
\end{aligned}
\end{equation}

\subsection{A good event and localization}\label{sec:event}

Throughout, $\epsilon\in(0,1/5)$ and $\delta\in(0,1)$.

\begin{claim}\label{claim:sample}
Let $\event$ denote the event that
\begin{enumerate}[label=\textup{(E\arabic*)},nosep]
    \item \label{item:E1} $|\hat w_i - w_i| \le \epsilon w_i$ for all $i\in[k]$,
    \item \label{item:E2} $\norm{\Delta_i}^2 \le \epsilon^2 V_i$ for all $i\in[k]$.
\end{enumerate}
For a uniform sample $Q$ of (sufficiently large) size $s=O\!\big(\tfrac{k^2}{\beta\epsilon^2\delta}\big)$ it holds that $\Pr[\event]\ge1-\delta/3$.
\end{claim}
\begin{proof}
We bound $\Pr[\lnot\ref{item:E1}]\le\delta/6$ and
$\Pr[\ref{item:E1}\wedge\lnot\ref{item:E2}]\le\delta/6$.
For the former, $s_i\sim\mathrm{Bin}(s,w_i)$ with $w_i\ge\beta/k$, so multiplicative
Chernoff and a union bound give
$\le2k\exp(-\epsilon^2 s\beta/(3k))\le\delta/6$ once
$s\ge\frac{3k}{\beta\epsilon^2}\ln(12k/\delta)$.
For the latter, on \ref{item:E1} we have $s_i\ge s(1-\epsilon)\beta/k$; conditionally
on $s_1,\dots,s_k$ the set $Q_i$ is an i.i.d.\ uniform sample of size $s_i$ from
$X_i^*$, so $\E[\norm{\Delta_i}^2\mid s_i]=V_i/s_i$ and Markov gives per-cluster
failure $\le 1/(s_i\epsilon^2)$. A union bound gives
$\le k/(s_i\epsilon^2)\le k^2/(s(1-\epsilon)\beta\epsilon^2)\le\delta/6$.
\end{proof}

Fix a constant $\gamma \geq 1$ and set
$\sol(Q)=\{C:\cost(C,Q)\le\gamma \OPTQ \}$, $\sol(P)=\{C:\cost(C,P)\le\gamma\OPT\}$.
Since $\OPT (Q) \leq \OPTQ$, we have $\{C:\cost(C,Q)\le\gamma \OPT(Q)\} \subseteq \sol (Q)$, so proving a guarantee uniformly over $\sol(Q)$ yields it over the $\OPT(Q)$-class as well.

The next lemma shows that every candidate of interest is pinned near the optimal
centers. It is stated for a generic a-priori bound $\cost(C^*,Q)\le\lambda\OPT$,
where $\lambda\ge1$ will be determined later as it depends on the exact application;
note that it implies $\OPTQ\le\lambda\OPT$ by \eqref{eq:optq-split}. Throughout we write
\begin{equation}\label{eq:Gam-def}
  \Gams^2\eqdef\frac{6\gamma\lambda k}{\beta}\,\OPT ,
  \qquad
  K\eqdef 11\gamma\lambda\OPT ,
\end{equation}
\begin{equation}\label{eq:Vset}
  \Vset\eqdef\Bigl\{v=(v_1,\dots,v_k)\in\Xspace^k: \forall i\in[k] \
    \norm{v_i}\le\Gams \textit{, and }
    \textstyle\sum_i w_i\norm{v_i}^2\le K\Bigr\}.
\end{equation}
We stress that $\Gams$, $K$, and hence the set $\Vset$, are determined by
$P,k,\beta,\gamma$ and $\lambda$ alone, and in particular do not depend on the sample $Q$. 

We start with a geometric claim bounding the boundary gain.

\begin{claim}\label{claim:farpoints}
Assume $\Gams\le D/4$. Then for every $v$ with $\norm{v_i}\le\Gams$ for all $i$, and
every $p\in P$, if $b_v(p)>0$ then $p\in F$, and
$0\le b_v(p)\le 4\Gams\,a(p)\le 4a(p)^2 .$
\end{claim}
\begin{proof}
Let $i=i(p)$ and let $j$ attain $\min_l\norm{p-\mu_l-v_l}$. If $b_v(p)>0$ then
$j\ne i$, and
$\norm{p-\mu_j}-\Gams\le\norm{p-\mu_j-v_j}\le\norm{p-\mu_i-v_i}\le a(p)+\Gams$, so
$\norm{p-\mu_j}\le a(p)+2\Gams$ and $D\le\norm{\mu_i-\mu_j}\le2a(p)+2\Gams$. Hence
$a(p)\ge D/2-\Gams\ge D/4$ (using $\Gams\le D/4$), i.e.\ $p\in F$ and $a(p)\ge\Gams$.
Writing $a=a(p)$ and $a_2=\min_{l\ne i}\norm{p-\mu_l}\ge a$ (as $p\in X_i^*$),
\[
  b_v(p)\le\norm{p-\mu_i-v_i}^2-\norm{p-\mu_j-v_j}^2
  \le(a+\Gams)^2-(a_2-\Gams)^2\le(a+\Gams)^2-(a-\Gams)^2=4\Gams a \leq 4 a^2,
\]
as $a_2 \geq a \geq \Gams$.
\end{proof}

In particular, averaging \cref{claim:farpoints} over $P$ and over $Q$ gives
\begin{equation}\label{eq:B-trivial}
  B_v(P)\le  4\OPT,\qquad B_v(Q)\le 4\cost(C^*,Q)\le 4\lambda\OPT .
\end{equation}

We next prove the localization lemma.

\begin{lemma}\label{lemma:sigma}
Suppose $\cost(C^*,Q)\le\lambda\OPT$ and $\Gams\le D/4$. Then, conditioned on
$\event$, every $C\in\union$ satisfies $C=C(v)$ for some $v\in\Vset$.
\end{lemma}

\begin{proof}
Fix $i\in[k]$ and let $R_i = \dist(\mu_i, C)$ and $c(p)\in C$ be a nearest center to $p$. First we show that for every $C \in \union$, for all $i \in [k]$, $R_i \leq \Gams$.
For every $p\in X_i^*$ we have, 
\[
  R_i^2\le\norm{\mu_i-c(p)}^2\le 2\norm{p-\mu_i}^2+2\norm{p-c(p)}^2
  =2a(p)^2+2\dist^2(p,C).
\]
Averaging over $p\in X_i^*$ and using
$\frac1{|X_i^*|}\sum_{p\in X_i^*}\dist^2(p,C)\le\frac1{w_i}\cost(C,P)$ gives
\begin{equation}\label{eq:loc-avg}
  R_i^2\ \le\ 2V_i+\frac{2}{w_i}\cost(C,P).
\end{equation}
If $C\in\sol(P)$, using $V_i\le\OPT/w_i$, $\cost(C,P)\le\gamma\OPT$, $\gamma\ge1$
and $w_i\ge\beta/k$,
\[
  R_i^2\le\frac{2(1+\gamma)\OPT}{w_i}\le\frac{4\gamma k}{\beta}\OPT\le\Gams^2 .
\]
Else, if $C\in\sol(Q)$ the same averaging over $Q_i$, with $\hat w_i,\hat V_i,\smean_i$
in place of $w_i,V_i,\mu_i$, gives
$\widehat R_i^2\le 2\hat V_i+\frac{2}{\hat w_i}\cost(C,Q)$ for
$\widehat R_i = \dist(\hat{\mu}_i, C)$. Using $\hat V_i\le\OPTQ/\hat w_i$,
$\cost(C,Q)\le\gamma\OPTQ$, $\OPTQ\le\lambda\OPT$, and
$\hat w_i\ge(1-\epsilon)\beta/k\ge\frac{4\beta}{5k}$ by \ref{item:E1},
\[
  \widehat R_i^2\le\frac{2(1+\gamma)\lambda\OPT}{\hat w_i}
  \le\frac{5\gamma\lambda k}{\beta}\OPT.
\]

By \ref{item:E2}, $\norm{\Delta_i}\le\epsilon\sqrt{V_i}\le\frac15\sqrt{k\OPT/\beta}$,
so, by the triangle inequality 
$$
R_i^2 \le (\widehat R_i+\norm{\Delta_i})^2 \le(\sqrt{5\gamma\lambda}+\tfrac15)^2 \cdot  \frac{k\OPT}{\beta}
\leq \frac{6\gamma\lambda k}{\beta}\OPT=\Gams^2,$$
where the last inequality follows from $\gamma \lambda \geq 1$.

Second, we show that for every $C \in \union$, $\sum_i w_i\norm{v_i}^2\le K$. As $\Gams\le D/4$, the balls $\{c:\norm{c-\mu_i}\le\Gams\}$, $i\in[k]$, are pairwise
disjoint; each contains at least one center of $C$ by the above, and $|C|=k$, so each
contains exactly one. Letting $\sigma(i)$ be the index of the unique center of $C$
inside the $i$-th ball and $v_i\eqdef c_{\sigma(i)}-\mu_i$ yields $C=C(v)$ with
$\norm{v_i}\le\Gams$ for all $i$.

For $C \in \sol(P)$, $\sum_i w_i\norm{v_i}^2=\cost(C(v),P)-\OPT+B_v(P)\le(\gamma+3)\OPT\le K$, using \eqref{eq:excess} for the first equality and \eqref{eq:B-trivial} for the second inequality.

For $C\in\sol(Q)$, the second line of \eqref{eq:excess} together with
$\cost(C,Q)\le\gamma\OPTQ$ and \eqref{eq:B-trivial} gives
$$\sum_i\hat w_i\norm{v_i-\Delta_i}^2\le(\gamma-1)\OPTQ+4\lambda\OPT\le(\gamma+3)\lambda\OPT.$$
Hence, by $\norm{v_i}^2\le2\norm{v_i-\Delta_i}^2+2\norm{\Delta_i}^2$, \ref{item:E2}
and $\epsilon\le1/5$,
\[
  \textstyle\sum_i\hat w_i\norm{v_i}^2
  \le 2(\gamma+3)\lambda\OPT+2(1+\epsilon)\epsilon^2\OPT
  \le 2(\gamma+3)\lambda\OPT+\tfrac1{10}\OPT ,
\]
and by \ref{item:E1}, $w_i\le\hat w_i/(1-\epsilon)\le\frac54\hat w_i$, so
$\sum_i w_i\norm{v_i}^2\le\frac52(\gamma+3)\lambda\OPT+\frac18\OPT\le 11\gamma\lambda\OPT=K$,
using $\gamma+3\le4\gamma$.
\end{proof}

\subsection{Matched-cost error}\label{sec:matched}

We next bound the terms of \eqref{eq:matched-expand}, uniformly over $\Vset$.

\begin{lemma}\label{lem:disp}
Assume $\Gams\le D/4$ and condition on $\event$. Then for every $v\in\Vset$:
\begin{enumerate}[label=\textup{(\roman*)},nosep]
  \item $\sum_i w_i\norm{v_i}^2\le4\cost(C(v),P)$, and, $\bigl|\sum_i(\hat w_i-w_i)\norm{v_i}^2\bigr|\le4\epsilon\cost(C(v),P)$,
  \item $\bigl|2\sum_i\hat w_i\ip{v_i}{\Delta_i}\bigr|\le5(1+\epsilon)\epsilon\cost(C(v),P)$,
  \item $\bigl|\sum_i\hat w_i\norm{\Delta_i}^2\bigr|\le(1+\epsilon)\epsilon^2\OPT$.
\end{enumerate}
\end{lemma}
\begin{proof}
(i): by \eqref{eq:excess} and \eqref{eq:B-trivial},
$\sum_i w_i\norm{v_i}^2=\cost(C(v),P)+B_v(P)-\OPT\le\cost(C(v),P)+3\OPT\le4\cost(C(v),P)$,
using $\OPT\le\cost(C(v),P)$. Then, by \ref{item:E1},
$\bigl|\sum_i(\hat w_i-w_i)\norm{v_i}^2\bigr|\le\epsilon\sum_i w_i\norm{v_i}^2\le4\epsilon\cost(C(v),P)$.
(ii): by Cauchy--Schwarz and Young inequalities,
$\bigl|2\sum_i\hat w_i\ip{v_i}{\Delta_i}\bigr| \le\sum_i\hat w_i(\epsilon\norm{v_i}^2+\epsilon^{-1}\norm{\Delta_i}^2)$. By
\ref{item:E1}, the first part is bounded by $\epsilon(1+\epsilon)\sum_i w_i\norm{v_i}^2\le4\epsilon(1+\epsilon)\cost(C(v),P)$. By \ref{item:E1},\ref{item:E2}, the second part is at most
$\epsilon^{-1}(1+\epsilon)\sum_i w_i\epsilon^2V_i=(1+\epsilon)\epsilon\OPT
\le(1+\epsilon)\epsilon\cost(C(v),P)$, so their sum is at most $5(1+\epsilon)\epsilon\cost(C(v),P)$.
(iii): by \ref{item:E1},\ref{item:E2},
$\bigl|\sum_i\hat w_i\norm{\Delta_i}^2\bigr|\le(1+\epsilon)\sum_i w_i\epsilon^2V_i=(1+\epsilon)\epsilon^2\OPT$.
\end{proof}

\subsection{Boundary-gain error}\label{sec:boundary}

We now bound $\abs{B_v(Q)-B_v(P)}$ uniformly over $v\in\Vset$, which is the main
technical difficulty of the proof. The obstacle is that $b_v(p)$ is nonlinear in the
displacements $v$: it is a maximum of functions that are affine in $v$ except for a
quadratic term $\norm{v_i}^2-\norm{v_j}^2$. The proof handles the two features
separately. Maurer's vector contraction inequality removes the maximum, and a
one-coordinate lifting turns the quadratic term into a linear one, after which the
whole process is a single inner product in a Hilbert space and Cauchy--Schwarz
finishes the argument.

\begin{theorem}[Maurer's vector contraction inequality {\citep[Theorem 2]{DBLP:conf/alt/Maurer16}}]
\label{thm:contraction}
Let $\mathcal S$ be a countable set and let $\psi_t:\mathcal S\to\mathbb R$,
$\phi_t:\mathcal S\to\ell_2$ $(1\le t\le m)$ satisfy
$\psi_t(\sigma)-\psi_t(\sigma')\le\norm{\phi_t(\sigma)-\phi_t(\sigma')}$ for all
$\sigma,\sigma'\in\mathcal S$. Let $\xi_t$ and $\eta_{t,j}$ be independent Rademacher
variables. Then
\[
  \E_\xi\sup_{\sigma\in\mathcal S}\sum_t\xi_t\,\psi_t(\sigma)
  \ \le\ \sqrt2\,\E_\eta\sup_{\sigma\in\mathcal S}\sum_t\sum_j\eta_{t,j}\,\phi_t(\sigma)_j ,
\]
where $\phi_t(\sigma)_j$ denotes the $j$-th coordinate of $\phi_t(\sigma)$.
\end{theorem}

\begin{lemma}\label{lem:boundary-mean}
Assume $\Gams\le D/4$ and let
$\;Z\eqdef\sup_{v\in\Vset}\bigl|B_v(Q)-B_v(P)\bigr|$. Then
\begin{equation*}
  \E_Q\,Z\;\le\;63\,\sqrt k\,\Gams\sqrt{T/s}.
\end{equation*}
\end{lemma}
\begin{proof}
Fix $p\in X_i^*$. By \cref{claim:farpoints}, $b_v(p)>0$ only when $p\in F$, so assume
$p\in F$. With $A(p)=\{j:\norm{p-\mu_j}\le a(p)+2\Gams\}$ we may write
\begin{equation}\label{eq:bv-max}
  b_v(p)=\max_{j\in A(p)}\bigl(r_j(p)+h_j(p,v)\bigr),
\end{equation}
where $r_j(p)=\norm{p-\mu_i}^2-\norm{p-\mu_j}^2\le0$ and
$h_j(p,v)=-2\ip{p-\mu_i}{v_i}+2\ip{p-\mu_j}{v_j}+\norm{v_i}^2-\norm{v_j}^2$. The
restriction to $A(p)$ is valid since for $j\notin A(p)$ the argument of
\cref{claim:farpoints} gives $r_j+h_j<0$, while $j=i\in A(p)$ attains $0$.

\emph{Symmetrization.} Let $Q'=(q_1',\dots,q_s')$ be an independent i.i.d.\ copy and
$\xi_t$ Rademacher. Since $B_v(P)=\E_{Q'}[\frac1s\sum_t b_v(q_t')]$ for every fixed
$v$, and the index set $\Vset$ depends on neither $Q$ nor $Q'$, Jensen applied to
$\sup_{v\in\Vset}$ and exchangeability of $(q_t,q_t')$ give
\begin{align*}
  &\E_Q Z
  =\frac1s \E_Q \sup_{v\in\Vset}\Bigl|\E_{Q'}\bigl[\textstyle\sum_t\big(b_v(q_t)-b_v(q_t')\big)\bigr]\Bigr|
  \ \leq\ \frac1s \E_{Q,Q'} \sup_{v\in\Vset}\Bigl|\sum_t\big(b_v(q_t)-b_v(q_t')\big)\Bigr| \\
  &= \frac1s \E_{Q,Q',\xi} \sup_{v\in\Vset}\Bigl|\sum_t\xi_t\big(b_v(q_t)-b_v(q_t')\big)\Bigr|
  \ \le\ \frac2s\,\E_{Q,\xi}\sup_{v\in\Vset}\Bigl|\sum_t\xi_t\,b_v(q_t)\Bigr| ,
\end{align*}
the sum running over far points $q_t\in Q\cap F$ only, as $b_v$ vanishes off $F$.

\emph{Contraction.} Fix $Q$ with far points $\{q_1,\dots,q_m\}=Q\cap F$, and write
$a_t=a(q_t)\ge D/4$, $i_t=i(q_t)$. By \eqref{eq:bv-max}, $v\mapsto b_v(q_t)$ is
$1$-Lipschitz with respect to the $\ell_\infty$, hence the $\ell_2$, norm of the
vector $\phi_t(v)=(h_j(q_t,v))_{j\in A(q_t)}$. Since $0\in\Vset$ and $b_0\equiv0$ we
have $\sup_{v\in\Vset}\sum_t\zeta_t b_v(q_t)\ge0$ for every sign pattern $\zeta$, so
writing $|x|=\max(x,-x)$ and using the symmetry of $\xi$, \cref{thm:contraction} gives
\begin{equation}\label{eq:after-contraction}
  \E_\xi\sup_{v\in\Vset}\Bigl|\sum_t\xi_t b_v(q_t)\Bigr|
  \le 2\sqrt2\,\E_\eta\sup_{v\in\Vset}\sum_t\sum_{j\in A(q_t)}\eta_{t,j}\,h_j(q_t,v).
\end{equation}

\emph{Lifting.} The functions $h_j(p,\cdot)$ are not linear in $v$, but they all
become linear after appending one coordinate to each displacement. Let
$\Gspace\eqdef\bigoplus_{r=1}^k(\Xspace\oplus\mathbb R)$ with the direct-sum inner
product, and define the \emph{lifted displacement}
\[
  \lift(v)\eqdef\bigl(\lift_1(v),\dots,\lift_k(v)\bigr)\in\Gspace ,
  \qquad
  \lift_r(v)\eqdef\Bigl(v_r,\ \frac{\norm{v_r}^2}{2\Gams}\Bigr) .
\]
For $p\in X_i^*$ and $j\in[k]$, let $g_{p,j}\in\Gspace$ be the vector whose only
nonzero blocks are
\[
  (g_{p,j})_i=\bigl(-2(p-\mu_i),\ 2\Gams\bigr),
  \qquad
  (g_{p,j})_j=\bigl(2(p-\mu_j),\ -2\Gams\bigr) ;
\]
for $j=i$ these two expressions cancel and $g_{p,i}=0$, consistently with
$h_i(p,\cdot)\equiv0$. Expanding the inner products gives the identity
\begin{equation}\label{eq:lift}
  h_j(p,v)=\ip{g_{p,j}}{\lift(v)}_{\Gspace}\qquad\text{for all }j\in[k],\ v\in\Xspace^k .
\end{equation}
Consequently the Rademacher process in \eqref{eq:after-contraction} is a single inner product. 

Setting $W\eqdef\sum_t\sum_{j\in A(q_t)}\eta_{t,j}\,g_{q_t,j}\in\Gspace$,
\begin{equation}\label{eq:CS-lift}
  \sup_{v\in\Vset}\sum_t\sum_{j\in A(q_t)}\eta_{t,j}h_j(q_t,v)
  =\sup_{v\in\Vset}\ip{W}{\lift(v)}
  \ \le\ \norm{W}\cdot U ,
  \qquad U\eqdef\sup_{v\in\Vset}\norm{\lift(v)} .
\end{equation}

\emph{Bounding $U$.} Since $\norm{v_r}\le\Gams$ we have
$\norm{\lift_r(v)}^2=\norm{v_r}^2+\frac{\norm{v_r}^4}{4\Gams^2}\le\frac54\norm{v_r}^2$,
and by the aggregate constraint in \eqref{eq:Vset} together with $w_r\ge\beta/k$,
\[
  \sum_r\norm{v_r}^2\le\frac{k}{\beta}\sum_r w_r\norm{v_r}^2\le\frac{kK}{\beta}
  =\frac{11\gamma\lambda k}{\beta}\OPT=\frac{11}{6}\,\Gams^2 .
\]
Hence $U^2\le\frac54\cdot\frac{11}{6}\Gams^2=\frac{55}{24}\Gams^2$. 

\emph{Bounding $\norm{W}$.} For $j\in A(q_t)$ we have $\norm{q_t-\mu_j}\le a_t+2\Gams\le3a_t$
and $\Gams\le D/4\le a_t$, so
\[
  \norm{g_{q_t,j}}^2=4\norm{q_t-\mu_{i_t}}^2+4\Gams^2+4\norm{q_t-\mu_j}^2+4\Gams^2
  \le 4a_t^2+4a_t^2+36a_t^2+4a_t^2=48\,a_t^2 .
\]
The $\eta_{t,j}$ are independent with zero mean, so
$\E_\eta\norm{W}^2=\sum_t\sum_{j\in A(q_t)}\norm{g_{q_t,j}}^2\le48\sum_t|A(q_t)|\,a_t^2\le48k\sum_t a_t^2$,
and by Jensen $\E_\eta\norm{W}\le(48k\sum_t a_t^2)^{1/2}$. Combining with
\eqref{eq:CS-lift} and $\sqrt{55/24}\cdot\sqrt{48}=\sqrt{110}<11$,
\[
  \E_\eta\sup_{v\in\Vset}\sum_t\sum_{j\in A(q_t)}\eta_{t,j}h_j(q_t,v)
  \ \le\ \sqrt{110\,k}\,\Gams\Bigl(\textstyle\sum_t a_t^2\Bigr)^{1/2}
  \ \le\ 11\sqrt k\,\Gams\Bigl(\textstyle\sum_{q\in Q\cap F}a(q)^2\Bigr)^{1/2}.
\]
Finally, combining with \eqref{eq:after-contraction} and using Jensen (concavity of
$\sqrt{\cdot}$) together with $\E_Q\sum_{q\in Q\cap F}a(q)^2=sT$,
\[
  \E_Q Z\le\frac2s\cdot2\sqrt2\cdot11\sqrt k\,\Gams(sT)^{1/2}
  \le63\,\sqrt k\,\Gams\sqrt{T/s}. \qedhere
\]
\end{proof}

% ============================================================================
\subsection{Stable coreset}\label{sec:stable}
% ============================================================================

We further augment $\event$ with a one-sided control on the sampled optimal cost,
i.e.\ the event $\text{(E3):}\ \cost(C^*,Q)\le(6/\delta)\OPT.$
Since $\E[\cost(C^*,Q)]=\OPT$, Markov gives $\Pr[\lnot\text{(E3)}]\le\delta/6$. In this
section we write $\event$ for $\ref{item:E1}\wedge\ref{item:E2}\wedge\text{(E3)}$,
so that \cref{claim:sample} gives $\Pr[\event]\ge1-\delta/2$.

We instantiate \eqref{eq:Gam-def} with $\lambda=6/\delta$, so
$\Gams^2=36\gamma k\OPT/(\beta\delta)$. Assume $\rho^2\le\beta\delta/(576\gamma k)$;
then $\Gams^2\le\OPT/(16\rho^2)\le D^2/16$ by \cref{lemma:rabani}, so $\Gams\le D/4$
and \cref{claim:farpoints,lemma:sigma,lem:disp,lem:boundary-mean} all apply.

\begin{lemma}\label{lem:boundary-stable}
Let $\rho^2\le\beta\delta/(576\gamma k)$. Then for a uniform sample of (sufficiently
large) size $O\!\big(k^2\gamma/(\beta\epsilon^2\delta^2)\big)$, with probability $\ge1-\delta/2$,
$\ \sup_{v\in\Vset}\abs{B_v(Q)-B_v(P)}\le\epsilon\OPT$.
\end{lemma}
\begin{proof}
By \cref{lem:boundary-mean} and the trivial bound $T\le\frac1n\sum_p a(p)^2=\OPT$,
\[
  \E_Q Z\le63\sqrt k\,\Gams\sqrt{\OPT/s},\qquad
  (\E_Q Z)^2\le3969\,k\Gams^2\OPT/s .
\]
For the variance we apply Efron--Stein to $Z$, which is a function of the i.i.d.\
coordinates $q_1,\dots,q_s$ only, the class $\Vset$ being fixed. By
\cref{claim:farpoints}, $0\le b_v(p)\le4\Gams\,a(p)$ for every $v\in\Vset$, so
replacing $q_t$ by an independent copy $q_t'$ changes $B_v(Q)$ by at most
$\frac{4\Gams}{s}(a(q_t)+a(q_t'))$ \emph{uniformly in} $v\in\Vset$, and hence changes
the supremum $Z$ by at most the same amount. By Efron--Stein and $\E[a^2]=\OPT$,
\[
  \Var(Z)\le\frac12\sum_t\E\bigl[(Z-Z^{(t)})^2\bigr]
  \le\frac{8\Gams^2}{s^2}\sum_t\E[(a(q_t)+a(q_t'))^2]\le\frac{32\Gams^2\OPT}{s}.
\]
Hence $\E[Z^2]\le4001\,k\Gams^2\OPT/s$, and by Markov
$\Pr[Z>\epsilon\OPT]\le\frac{4001k\Gams^2}{s\epsilon^2\OPT}\le\delta/2$
once $s\ge8002\,k\Gams^2/(\epsilon^2\OPT\delta)=\Theta\!\big(k^2\gamma/(\beta\epsilon^2\delta^2)\big)$.
\end{proof}

\begin{lemma}\label{lemma:k-rcda}
On $\event$ and the event of \cref{lem:boundary-stable}, for every $C\in\union$,
\[
  \Bigl|(\cost(C,P)-\OPT)-(\cost(C,Q)-\cost(C^*,Q))\Bigr|\le 11\epsilon\cost(C,P).
\]
That is, $Q$ is an $11\epsilon$-RCDA of $P$ over $\union$.
\end{lemma}
\begin{proof}
On $\event$, \cref{lemma:sigma} gives $C=C(v)$ for some $v\in\Vset$, so it suffices to
bound $|E(v)|$ from \eqref{eq:rcda-error} uniformly over $v\in\Vset$. Bounding the
first two terms of \eqref{eq:rcda-error} by items (i) and (ii) of \cref{lem:disp}, and
the third by \cref{lem:boundary-stable} together with $\OPT\le\cost(C(v),P)$, we get,
using $\epsilon\le1/5$,
\[
  \abs{E(v)}\le4\epsilon\cost(C(v),P)+5(1+\epsilon)\epsilon\cost(C(v),P)+\epsilon\cost(C(v),P)
  \le 11\epsilon\cost(C(v),P). \qedhere
\]
\end{proof}

\begin{proof}[Proof of \cref{thm:kmeans-stable}]
By a union bound, $\event$ and the event of \cref{lem:boundary-stable} occur
simultaneously with probability at least $1-\delta$. On $\event$ we have
$\cost(C^*,Q)\le6\cost(C^*,P)/\delta$, so combining \cref{lemma:k-rcda} with
\cref{thm:rcda-to-stable} gives the claim; we use the restricted-class version of
\cref{thm:rcda-to-stable}, whose proof is identical after restricting every candidate
center set appearing in it to $\union$. Conclude by rescaling $\epsilon$.
\end{proof}

\subsection{Strong coreset}\label{sec:strong}

Fix $\epsilon \in (0,1/5)$. To obtain a strong coreset guarantee over $\union$ we require the additional assumption of bounded kurtosis.

\begin{definition}\label{def:kurtosis}
$P\subseteq\Xspace$ has $M$-bounded kurtosis if for every $i$, $\frac1{|X_i^*|}\sum_{p\in X_i^*}\norm{p-\mu_i}^4\le MV_i^2$.
\end{definition}

\begin{claim}\label{lem:fourth}
Under $M$-bounded kurtosis, $\frac1n\sum_{p\in P}a(p)^4\le\frac{Mk}{\beta}\OPT^2$.
\end{claim}
\begin{proof}
$\frac1n\sum_p a(p)^4=\sum_i w_i\frac1{|X_i^*|}\sum_{p\in X_i^*}\norm{p-\mu_i}^4\le M\sum_i w_iV_i^2$, and $V_i\le\OPT/w_i\le k\OPT/\beta$ gives $\sum_i w_iV_i^2\le\frac{k}{\beta}\OPT^2$.
\end{proof}

We then bound the cost of the far points.
\begin{claim}\label{lem:tail}
$T\le16\rho^2\frac{Mk}{\beta}\OPT$.
\end{claim}
\begin{proof}
$\indic{p\in F}\le(4a(p)/D)^2$ gives $T\le\frac{16}{D^2}\cdot\frac1n\sum_p a(p)^4$. Proof is concluded by applying \cref{lem:fourth} and $D^2\ge\OPT/\rho^2$ (\cref{lemma:rabani}).
\end{proof}

Consider the event $\text{(}\Bar{\text{E}}3\text{)}$: $\bigl|\cost(C^*,Q)-\OPT\bigr|\le\epsilon\OPT$.
\begin{claim}\label{claim:E3strong}
Under $M$-bounded kurtosis, if $s\ge6Mk/(\beta\epsilon^2\delta)$ then $\Pr[\lnot\text{(}\Bar{\text{E}}3\text{)}]\le\delta/6$.
\end{claim}
\begin{proof}
$\E[\cost(C^*,Q)]=\OPT$ and, by \cref{lem:fourth}, $\Var(\cost(C^*,Q))\le\E[a^4]/s\le Mk\OPT^2/(\beta s)$.
Using Chebyshev's inequality the probability of failure is at most $\frac{Mk}{s\beta\epsilon^2}\le\delta/6$.
\end{proof}

Augment $\event$ with $\text{(}\Bar{\text{E}}3\text{)}$, then $\Pr[\event]\ge1-\delta/2$.

\begin{corollary}\label{cor:optbound}
On $\event$, $\abs{\OPTQ-\OPT}\le2\epsilon\OPT$.
\end{corollary}
\begin{proof}
Using \eqref{eq:optq-split}, $\cost(C^*,Q)=\OPTQ+\sum_i\hat w_i\norm{\Delta_i}^2$. For the first direction,
$\OPTQ\le\cost(C^*,Q)\le(1+\epsilon)\OPT$. For the second direction,
\[
\OPTQ = \cost(C^*,Q) - \sum_i\hat w_i\norm{\Delta_i}^2 \geq (1-\epsilon)\OPT - \sum_i (1+\epsilon) w_i \epsilon^2 V_i \geq (1-2\epsilon)\OPT. \qedhere
\]
\end{proof}

On $\event$ we have $\cost(C^*,Q)\le(1+\epsilon)\OPT$, so we instantiate
\eqref{eq:Gam-def} with $\lambda=1.2$, giving $\Gams^2=7.2\gamma k\OPT/\beta$. Assume
$\rho^2\le\beta/(120\gamma k)$; then $\Gams^2\le\OPT/(16\rho^2)\le D^2/16$, so
$\Gams\le D/4$ and \cref{claim:farpoints,lemma:sigma,lem:disp,lem:boundary-mean} apply.
We first bound the matched cost.

\begin{lemma}\label{claim:matched}
On $\event$, for every $v\in\Vset$,
$\bigl|\costm(C(v),P)-\costm(C(v),Q)\bigr|\le 14\epsilon\cost(C(v),P)$.
\end{lemma}
\begin{proof}
Bound the four terms of \eqref{eq:matched-expand}: the first by \cref{cor:optbound},
the remaining three by items (i), (ii), (iii) of \cref{lem:disp}. Together with
$\OPT\le\cost(C(v),P)$ this gives
{\small
\begin{align*}
\bigl|\costm(C(v),P)-\costm(C(v),Q)\bigr| &\le 2\epsilon\OPT+4\epsilon\cost(C(v),P)+5(1+\epsilon)\epsilon\cost(C(v),P)+(1+\epsilon)\epsilon^2\OPT \\ &\le 14 \epsilon\cost(C(v),P),
\end{align*}
}
the last step using $\epsilon\le1/5$.
\end{proof}

We next bound the boundary-gain error.

\begin{lemma}\label{lem:boundary-strong}
Under $M$-bounded kurtosis and $\rho^2\le\beta/(120\gamma k)$, for a sample of (sufficiently large) size
$O (k^2M/(\beta\epsilon^2\delta))$, with probability $\ge1-\delta/2$,
\[\sup_{v\in\Vset}\abs{B_v(Q)-B_v(P)}\le\epsilon\OPT.\]
\end{lemma}
\begin{proof}
By \cref{lem:boundary-mean}, $\E_Q Z\le63\sqrt k\,\Gams\sqrt{T/s}$. With
$\Gams^2=7.2\gamma k\OPT/\beta$, \cref{lem:tail} and $\rho^2\le\beta/(120\gamma k)$ give
$T\le\frac{2M}{15\gamma}\OPT$ and hence
$\Gams^2T\le\frac{7.2\gamma k\OPT}{\beta}\cdot\frac{2M\OPT}{15\gamma}\le\frac{kM}{\beta}\OPT^2$, so
\[
\E_Q Z\le63\sqrt{k\,\Gams^2T/s}\le63\,\OPT\sqrt{\frac{k^2M}{\beta s}},\qquad
(\E_Q Z)^2\le3969\,\frac{k^2M}{\beta s}\OPT^2 .
\]
For the variance, by \cref{claim:farpoints}, $0\le b_v(p)\le4a(p)^2$ for every
$v\in\Vset$, so replacing $q_t$ changes $B_v(Q)$ by at most
$\frac4s(a(q_t)^2+a(q_t')^2)$ uniformly in $v\in\Vset$, and hence changes $Z$ by at
most the same amount. By Efron--Stein and \cref{lem:fourth},
\[
\Var(Z)\le\frac12\sum_t\E[(Z-Z^{(t)})^2]\le\frac{32}{s}\E[a^4]\le\frac{32Mk}{\beta s}\OPT^2 .
\]
Hence $\E[Z^2]=(\E[Z])^2+\Var(Z)\le\frac{4001k^2M}{\beta s}\OPT^2$, and Markov gives
$\Pr[Z>\epsilon\OPT]\le\frac{4001k^2M}{\beta s\epsilon^2}\le\delta/2$.
\end{proof}

\begin{proof}[Proof of \cref{thm:kmeans-strong}]
By a union bound, $\event$ and the event of \cref{lem:boundary-strong} occur
simultaneously with probability at least $1-\delta$. By \cref{lemma:sigma}, every
$C\in\union$ equals $C(v)$ for some $v\in\Vset$, so applying the identity
\eqref{eq:master} together with \cref{claim:matched} and \cref{lem:boundary-strong},
\[\abs{\cost(C,P)-\cost(C,Q)}\le 14\epsilon\cost(C,P)+\epsilon \OPT \le 15 \epsilon\cost(C,P).\]
The theorem statement is obtained by rescaling $\epsilon$.
\end{proof}

\section{Immediate Applications of Theorem~\ref{thm:stable}}\label{sec:applications}

We derive immediate applications of Theorem~\ref{thm:stable} through the transfer-to-submetric property of stable coresets, we collect them here for potential future use.

\paragraph{Kernel Clustering.}
A kernel function $K \colon X \times X \to \RR$ is characterized by the existence of a
Hilbert space $\Xspace$ and a feature map $\varphi \colon X \to \Xspace$ such that
$K(x, y) = \ip{\varphi(x)}{\varphi(y)}_\Xspace$ for all $x, y \in X$.
In \emph{kernel $1$-mean}, the goal is to minimize
$\cost_\varphi(c, P) \eqdef \frac{1}{n}\sum_{x \in P}\norm{\varphi(x) - \varphi(c) }_\Xspace^2$
over all $c \in X$.

\begin{corollary}
Let $P \subseteq X$ be a finite set with kernel $K \colon X \times X \to \RR$. Then for every $\epsilon \in (0, 1/10)$ and $\delta >0$, a uniform sample of (sufficiently large) size $O(\epsilon^{-2}\delta^{-1})$ from $P$ is,
with probability at least $1-\delta$,
a stable $(\delta\epsilon, 8\epsilon)$-coreset for kernel $1$-mean with respect to $K$.
\end{corollary}

\paragraph{Clustering in $\ell_p$ for $0 < p \le 2$.}
\label{subsec:ellp}

For every $0 < p \le 2$ and every $0 < \alpha \le p/2$, the metric space
$(\ell_p^d, \norm{\cdot}_p^\alpha)$ embeds isometrically into a Hilbert space,
this follows from a classical theorem of \citet{schoenberg1938}
(see also~\citet{BDCK67} and~\citet[Theorem~4.10]{WellsWilliams1975}).
In addition to Corollary~\ref{cor:median}, taking $\alpha=p/2$ yields the following.

\begin{corollary}
Let $p \in (0, 2]$ and $P \subset \RR^d$ finite. Then for every $\epsilon \in (0, 1/10)$ and $\delta >0$, a uniform sample of (sufficiently large) size $O(\epsilon^{-2}\delta^{-1})$ from $P$ is,
with probability at least $1-\delta$, a stable
$(\delta\epsilon, 8\epsilon)$-coreset for the cost
\[
\cost(c, P) \;\eqdef\; \frac{1}{n}\sum_{x \in P} \norm{x - c}_p^{\,p}.
\]
\end{corollary}

\paragraph{Approximation algorithms for $k$-median.}
A known framework reduces the design of approximation algorithms for $k$-median to obtaining weak coreset guarantees for the $1$-median problem through uniform sampling~\citep{DBLP:conf/focs/KumarSS04,DBLP:journals/talg/AckermannBS10}.
This framework was recently instantiated by~\citep{carmel2025stablecoresetsunleashingpower} using stable coresets, yielding approximation algorithms for $k$-median in $\ell_1$ and every submetric of $\ell_1$, including Hamming, Kendall--tau, Jaccard, and tree metrics (see Theorem~5.4 in~\citep{carmel2025stablecoresetsunleashingpower}). Since the running time of the framework depends exponentially on the coreset size, our improved stable coreset bounds translate directly into faster approximation algorithms.

\end{document}